\documentclass[a4paper]{easychair}

\usepackage{booktabs}
\usepackage{subcaption}
\usepackage{latexsym}
\usepackage{setspace}
\usepackage{cancel}
\usepackage{listings}
\usepackage{pgfplots}
\usepackage{graphicx}
\usepackage{appendix}
\usepackage{stmaryrd}
\usepackage{leftidx}
\usepackage{mathtools}
\usepackage{paralist}
\usepackage{color}
\usepackage{mathrsfs}
\usepackage{tikz}
\usetikzlibrary{shapes,automata,positioning,decorations.markings}
\usepackage[linesnumbered,ruled,noend]{algorithm2e}
\usepackage{wrapfig}
\usepackage{amssymb,amsmath}
\usepackage{siunitx} 

\usepackage{diego-macro}

\newtheorem{example}{Example}[section]
\newtheorem{definition}[example]{Definition}
\newtheorem{proposition}[example]{Proposition}
\newtheorem{theorem}[example]{Theorem}
\newtheorem{lemma}[example]{Lemma}

\renewcommand{\phi}{\varphi}%

\newcommand{\struct}{\mathfrak{S}}

\newcommand{\vocab}{\sigma}

\newcommand{\defn}[1]{\textit{#1}}

\newcommand{\Sem}[1]{\ensuremath{[\![#1]\!]}}

\newcommand{\Aut}{\ensuremath{\mathcal{A}}} 
\newcommand{\transrel}{\ensuremath{\Delta}} 

\newcommand{\controls}{\ensuremath{Q}} 

\newcommand{\params}{\mathcal{X}}
\newcommand{\locals}{\mathcal{Y}}
\newcommand{\nonterm}{\mathcal{N}}

\newcommand{\Grammar}{\ensuremath{\mathcal{G}}} 

\newcommand{\curr}{{\textit{curr}}}
\newcommand{\stktop}{{\textit{top}}}

\newcommand{\OMIT}[1]{}

\newcommand{\ialphabet}{\Sigma}

\newcommand{\theory}{T}
\newcommand{\Constants}{\mathcal{C}}
\newcommand{\Variables}{\mathcal{V}}
\newcommand{\iparam}{X}
\newcommand{\ivar}{Y}
\newcommand{\outfmla}{\Phi}

\newcommand\mathexpl[2]{\left[\text{\begin{minipage}{#1}{\it\small #2}\end{minipage}}\right]}

\newcommand\defeq{:=}
\newcommand\synalt{\ |\ }

\newcommand\svar{x}
\newcommand\varsvar{y}
\newcommand\sexp{s}
\newcommand\word{w}
\newcommand\iexp{i}
\newcommand\regex{r}
\newcommand\aut{\Aut}
\newcommand\cha{a}
\newcommand\chb{b}
\newcommand\preds{\Psi}
\newcommand\pred{\psi}
\newcommand\tpred{p}
\newcommand\idxi{i}
\newcommand\idxj{j}
\newcommand\scount{c}
\newcommand\chacount{k}

\newcommand\rrange{\mathit{re.in\_range}}
\newcommand\sconcat{\mathit{str.\mathord{+}\mathord{+}}}
\newcommand\scontains{\mathit{str.contains}}
\newcommand\sinre{\mathit{str.in\_re}}
\newcommand\slen{\mathit{str.len}}
\newcommand\sprefixof{\mathit{str.prefixof}}
\newcommand\sreplace{\mathit{str.replace}}
\newcommand\ssubstr{\mathit{str.substr}}
\newcommand\ssuffixof{\mathit{str.suffixof}}

\newcommand\AbstractSEXP{\mathrm{AbstractSEXP}}
\newcommand\Assert{\mathrm{Assert}}
\newcommand\anyparikh{\varphi_\top}

\newcommand\labels{\mathrm{Labels}}
\newcommand\translabelled{\mathrm{TransLabelled}}
\newcommand\atran{t}
\newcommand\countof[2]{|#2|_{#1}}
\newcommand\run{\rho}
\newcommand\bucket{B}

\newcommand\ComputeBuckets{\mathrm{ComputeBuckets}}
\newcommand\bucketlist{\mathrm{BList}}

\newcommand\cpp{C$\mathord{+}\mathord{+}$}

\definecolor{light-gray}{gray}{0.9}
\definecolor{light-yellow}{RGB}{255, 255, 220}

\usepackage[backgroundcolor=orange!50, textsize=scriptsize,textwidth=1cm]{todonotes}

\definecolor{light-gray}{gray}{0.9}
\definecolor{light-yellow}{RGB}{255, 255, 220}

\newcommand\popleasychair[2]{#2}

\title{Parikh's Theorem Made Symbolic\footnote{
    An Errata to the POPL 2024 version of this paper is provided in Section~\ref{sec:errata}.
}}

\author{
    Matthew Hague\inst{1}
    \and
    Artur Je\.z\inst{2}
    \and
    Anthony W.\ Lin\inst{3}
}

\institute{
    Royal Holloway University of London \\
    \email{matthew.hague@rhul.ac.uk}
    \and
    University of Wroclaw \\
    \email{aje@cs.uni.wroc.pl}
    \and
    University of Kaiserslautern and Max-Planck Institute \\
    \email{awlin@mpi-sws.org}
}

\authorrunning{M.\ Hague, A.\ Je\.z, A.\ W.\ Lin}
\titlerunning{Parikh's Theorem Made Symbolic}

\begin{document}

\maketitle

\begin{abstract}
Parikh's Theorem is a fundamental result in automata theory with numerous
applications in computer science. These include software verification (e.g.
    infinite-state
verification, string constraints, and theory of arrays), verification of
cryptographic protocols (e.g. using Horn clauses modulo equational theories)
and database querying (e.g. evaluating path-queries in graph databases), among
others. Parikh's Theorem states that the letter-counting abstraction of
a language recognized by finite automata or context-free grammars is definable
in Linear Integer Arithmetic (a.k.a. Presburger Arithmetic). In fact, there is a
linear-time algorithm computing existential Presburger formulas capturing
such abstractions, which enables an efficient analysis via SMT-solvers.
Unfortunately, real-world applications typically require large alphabets (e.g.
Unicode, containing a~million of characters) --- which are well-known to be not
amenable to explicit treatment of the alphabets --- or even worse infinite
alphabets.

Symbolic automata have proven in the last decade to be an effective algorithmic
framework for handling large finite or even infinite alphabets.
A symbolic automaton employs an effective boolean algebra, which offers a
symbolic representation of character sets (i.e. in terms of predicates) and
often lends itself to an exponentially more succinct representation of
a language. Instead of letter-counting, Parikh's Theorem for symbolic automata
amounts to counting the number of times different predicates are satisfied by
an input sequence. Unfortunately, naively applying Parikh's Theorem from
classical automata theory to symbolic automata yields existential
Presburger formulas of exponential size.
In this paper, we provide a new construction for Parikh's Theorem for
symbolic automata and grammars, which avoids this exponential blowup: our
algorithm computes an existential formula in polynomial-time over (quantifier-free) Presburger and the base theory. In fact,
our algorithm extends to the model of parametric symbolic grammars, which are
one of the most expressive models of languages over infinite alphabets.
We have implemented our algorithm and show it can be used to solve string
constraints that are difficult to solve by existing solvers.
\end{abstract}

\section{Introduction}
\label{sec:intro}

Parikh's Theorem \cite{parikh} (see also \cite[Chapter H]{kozen-book}) is a
celebrated result in automata theory with far-reaching
applications in computer science. These include software verification
\cite{EG11,HL11,HL12}, decision procedures for array and string theories
\cite{DHK16,LB16,ostrich-int,JT19,chain-free}, and
evaluation and optimization of database queries \cite{tods12,DLT12}, among others. Parikh's Theorem
concerns the so-called \emph{letter-counting abstractions} of strings and
languages. For example, the Parikh
image of the string $abaacb$ is the mapping $f: \{a,b,c\} \to \N$, where $f(a) =
3$, $f(b) = 2$, and $f(c) = 1$.
In other words, the Parikh mapping abstracts away the
ordering from a string (resp.\ a set of strings), i.e.\ yielding a multiset
(resp.\ a set of multisets). Parikh's Theorem states
that the class of context-free languages and the class of regular languages
coincide modulo a Parikh mapping, both of which are moreover expressible as a
formula in
Linear Integer Arithmetic (a.k.a.\ Presburger Arithmetic). This is illustrated
in the following example.
\begin{example}
The Parikh
image of the regular language $L := (ab)^*$ is the set $S$ containing all
mappings $f:
\{a,b\} \to \N$ with $f(a) = f(b)$. Observe that $S$ is also the Parikh image of
the context-free language $\{ a^n b^n : n \geq 0 \}$.
The Parikh image of $L$ can be expressed as $x_a = x_b
\wedge x_a \geq 0$, where $x_a$ (resp. $x_b$) represents the count for the
letter $a$ (resp. $b$).
\end{example}

Although the classical formulation of Parikh's Theorem concerns mainly the
expressiveness of language models modulo taking Parikh images, its usefulness
in applications was enabled only decades later by the development of efficient
algorithms that compute an existential LIA formula (i.e.\ of the
form $\exists \bar x\varphi$, where $\varphi$ is quantifier-free) from a given
automaton/grammar, enabling the exploitation of highly optimized SMT-solvers.
In fact, building on the result by Esparza \cite{Esparza97}, Verma et al.
\cite{VermaSS05} develops a linear-time algorithm that computes an existential
LIA formula capturing the Parikh image of a given grammar. These results enabled
the exploitation of Parikh's Theorem in many applications. Among others, these
include
verification of multithreaded programs with counters and possibly with
(recursive) function calls \cite{HL11,HL12,To09}, verification of
concurrent and multithreaded programs \cite{EG11,HL12}, verification of
cryptographic protocols \cite{VermaSS05}, decision procedures for array
theories \cite{DHK16}, decision procedures for string constraints
\cite{LB16,ostrich-int,JT19,chain-free}, query evaluation over graph databases
\cite{tods12}, and reasoning over XML documents \cite{DLT12}. The following
two examples illustrate two simple applications of Parikh's Theorem for
difficult problems.

\begin{example}
    The problem of checking emptiness of the intersection of several
    context-free languages has immediate applications in static analysis of
    concurrent programs (e.g.\ see \cite{BET03}). However, since the problem of
    checking emptiness of two context-free languages is well-known to be
    undecidable (e.g.\ see \cite{kozen-book}), multiple incomplete methods are
    proposed, which include Parikh abstractions \cite{BET03} and synthesis of
    regular separators/overapproximations
    \cite{covenant,covenant-journal,lcegar}, among others. Take the
    two context-free languages used in the benchmark of \cite{covenant-journal}:
    \[
        L_1 := \{ a^n c a^n : n > 0 \} \qquad L_2 := \{ a^n c b^n : n > 0 \}
    \]
    That $L_1 \cap L_2 = \emptyset$ can be shown rather easily by considering
    the Parikh images $S_1$ and $S_2$ of, respectively, $L_1$ and $L_2$. In
    fact, $S_1$ contains precisely all multisets $f$ with $c \mapsto 1$,
    $b \mapsto 0$, and $a \mapsto i$, where $i$ is positive even number. On the
    other hand, $S_2$ contains precisely all multisets $g$ with
    $c \mapsto 1$, $a,b \mapsto i$, where $i$ is a positive number. Thus,
    it follows that $S_1 \cap S_2 = \emptyset$. By employing the linear-time
    algorithm \cite{VermaSS05} for Parikh images of context-free grammars and
    SMT-solvers, that $S_1 \cap S_2 = \emptyset$ can be easily verified.
\end{example}

\begin{example}
    \label{ex:we}
    String constraint solving is an area that has received a lot of attention,
    owing to their applications in the symbolic execution programs, for example in JavaScript
    \cite{Saxena-JS,expose,popl22,aratha}.
    In this example, we deal with the simple string constraint
    \[
        \varphi\ ::=\ zyx = xxz \wedge x \in a^* \wedge y \in a^+b^+ \wedge z \in b^*,
    \]
    which the authors of a recent paper \cite{noodle} have found to lead to
    failure for all string solvers that they have tried. We want to find a solution
    (i.e.\ mapping from the string variables $x,y,z$ to strings over the alphabet
    $\ialphabet = \{a,b\}$) satisfying all the restrictions. The first
    restriction is an equation $zyx = xxz$, which enforces the two different
    concatenations of the strings instantiating the variables to be equal.
    For example, the mapping $\lambda$, where $\lambda: z \mapsto aa$ and
    $\lambda: y,x \mapsto a$, satisfies this restriction; whereas the mapping
    $\lambda$, where $\lambda: z \mapsto a$ and $\lambda: y,x \mapsto b$, does
    not. Each of the other restrictions is a \emph{regular constraint}, which
    enforces a solution of a variable to satisfy certain regular patterns. For
    example, $x \in a^*$ enforces that the variable $x$ should be instantiated
    to be a string consisting only the letter $a$.

    By using letter-counting abstraction, we can easily show the above example
    to be unsatisfiable. For each $l \in \ialphabet$, let $|x|_l$ denote the
    number of times $l$ appearing 
     in $x$. The letter-counting abstraction of the
    equation is the following quantifier-free LIA formula:
    $\bigwedge_{l \in \ialphabet} |x|_l + |y|_l + |z|_l = 2|x|_l + |z|_l$,
    which can be simplified to $\bigwedge_{l \in \ialphabet} |y|_l = |x|_l$.
    The letter-counting abstraction of the regular constraint $x \in a^*$
    (resp. $z \in b^*$) is $|x|_a \geq 0 \wedge |x|_b = 0$
    (resp. $|z|_a = 0 \wedge |z|_b \geq 0$). Finally, the letter-counting
    abstraction of the regular constraint $y \in a^+b^+$ is
    $|y|_a > 0 \wedge |y|_b > 0$. Therefore, the letter-counting abstraction of
    $\varphi$ is the quantifier-free LIA formula
    \[
        \left( \bigwedge_{l \in \ialphabet} |y|_l = |x|_l  \right) \wedge
        |x|_a \geq 0 \wedge |x|_b = 0 \wedge
        |y|_a > 0 \wedge |y|_b > 0 \wedge
        |z|_a = 0 \wedge |z|_b \geq 0.
    \]
    This formula is easily seen to be unsatisfiable since it asserts that
    $0 = |x|_b = |y|_b > 0$. Furthermore, this can be easily checked by
    virtually all existing SMT-solvers which support LIA (e.g.\ Z3 \cite{Z3}).
\end{example}

Despite the usefulness of Parikh's Theorem, most 
real-world applications
require either large finite or even worse infinite alphabets, which renders
the classical Parikh's Theorem impractical. For example,
a regex $r$ over UTF-16 
--- e.g.\ \texttt{([\string^\textbackslash x00-\textbackslash x7F][\string^\textbackslash
x00-\textbackslash x7F])+}, which accepts non-empty
strings over non-ASCII characters of even length --- has a total of $2^{16}$
characters. A direct application of the linear algorithm from Verma et al.
\cite{VermaSS05} would yield a LIA formula with at least $2^{16}$ variables,
each keeping track of the count for each letter in UTF-16.

\paragraph{Symbolic automata framework.} The framework of symbolic automata
\cite{DV21,symbolic-power,veanes12} (a.k.a.\ automata modulo theories) has
proven in the last decade to be a
fruitful approach for handling large finite or even infinite alphabets.
The key to the framework is the symbolic representation of alphabets known as
\emph{effective boolean algebras}. Loosely speaking, an effective boolean
algebra is a domain $D$ with a class of monadic predicates (i.e.\ each has
an interpretation as a subset of $D$) that is closed under boolean operations
(set-union, set-intersections, and set-complementation). The term ``effective''
refers to the fact that each monadic predicate $P$ describes a \emph{syntactic} property
(e.g.\ a~character class in Unicode, or a LIA formula with
one free variable), and that checking whether the interpretation $\Sem{P}
\subseteq D$ is empty is decidable. Many examples of effective boolean algebras
are available including, notably, SMT algebras. For example, a LIA boolean 
algebra consists
of domain $D = \Z$ and monadic predicates of LIA (with existential quantifiers
allowed), e.g., $P := x\equiv 0 \pmod{2}$.
The syntactic representation of predicates
$P$ and decidability of checking emptiness can be taken advantage of by allowing
an automaton transition of the form $p \to_P q$, where $p, q$ are two automata
states, representing all (potentially infinitely many) transitions of the
form $p \to_a q$, where $a$ satisfies $P$. Symbolic automata extend normal automata by
allowing such transitions. An analogous representation of symbolic automata
in terms of symbolic (regular) expressions \cite{DV21,SVB21} is also possible, 
where predicates are allowed instead of concrete letters, e.g.\ the expression 
$P^+$ represents the sequences of strings of odd numbers, whenever
$P := x \equiv 1 \pmod{2}$.

Most analysis of symbolic automata is known to be reducible to the case of
normal automata, but with an exponential blow-up in the alphabet size
and the number of transitions \cite{symbolic-power,veanes12}.
Although in most cases such an exponential blow-up is unavoidable in the worst
case, clever algorithms that circumvent this exponential blow-up in practice
have been devised on basic automata operations (e.g.\ boolean operations,
transductions, learning, etc.). This takes us to the question of Parikh's
Theorem in the setting of symbolic automata, which has so far not received much
attention in the literature of symbolic automata.

A natural counterpart of letter-counting abstractions in the framework of
symbolic automata is \emph{predicate-counting abstractions}. Let us revisit
Example \ref{ex:we} but with the letters $a$ and $b$ instantiated with different
character classes. Let us start with $a := \texttt{\textbackslash d}$ (meaning a
digit) and $b := \texttt{\textbackslash D}$ (meaning a non-digit). In this case,
the predicate-counting abstraction with respect to $a$ and $b$ simply counts
the
numbers of occurrences of digits and non-digits in each string instantiations
of $x,y,z$. The same reasoning used in Example \ref{ex:we} will allow us to
prove unsatisfiability. On the other hand, consider $a := \texttt{\textbackslash
s}$ (including space symbols, tabs, and newlines) and $b :=
\texttt{.}$ (meaning any character, except for a newline). Then, $a$ (resp. the
complement $\bar a$) and $b$ (resp.\ $\bar b$) have non-empty intersections.
(More precisely, $a \cap b$ contains a space symbol, $a \cap \bar b$ contains
a newline character, and $\bar a \cap b$ contains (say) a digit.)
In general,
$n$ predicates in a symbolic automaton (equivalently, symbolic regular
expression) can give rise to $O(2^n)$ different ``combinations'' (a.k.a.\
\emph{min-terms}). In other words,
predicate-counting abstractions over symbolic automata can be reduced to
letter-counting abstractions of normal automata, but \emph{over an exponentially
bigger alphabet}. Thus, the linear-time construction of Verma et al. for
Parikh images of automata/grammars with (say) 14 predicates would yield already
a large LIA formula with more than 15000 integer variables, which is very
challenging to solve for existing SMT-solvers.

\paragraph{Contributions.} Our main result is the first polynomial-time 
algorithm for computing an existential formula that captures the
predicate-counting abstraction of a given symbolic automaton. In fact, the
algorithm extends to more expressive formalisms, namely, symbolic context-free
grammars even when they are extended with ``read-only registers''; this is a
model referred to as \emph{parametric symbolic grammars}, which extend
both symbolic automata \cite{DV21,symbolic-power}, symbolic visibly pushdown 
automata \cite{DA14}, and symbolic variable/parametric automata 
\cite{GKS10,FK20,FJL22}. This new formalism has further applications including 
solving complex string constraints, e.g., with context-free constraints and, to
some extent, the infamously difficult operator \texttt{to\_re}, which converts
strings to regular expressions.
We have provided an implementation of our algorithm 
and demonstrated its efficacy in solving some difficult string constraints 
examples. We detail these contributions below.

As described above, the main technical difficulty of our problem is that the 
standard reduction from symbolic automata $\Aut$ to normal automata $\Aut'$ 
yields an exponential-sized 
alphabet, i.e.\ $2^n$ when counting $n$ predicates over theory $\theory$. This is in general not
avoidable, e.g., symbolic regular expressions of the form $P_1P_2\cdots P_n$ 
over the LIA algebra, where $P_i$ represents the set of all 
numbers that are congruent to 0 modulo the $i$th prime, have $O(2^n)$
feasible min-terms.
It turns out that, when considering predicate-counting abstractions,
if $w$ is accepted by $\Aut$, there is $w'$ that is also accepted by $\Aut$,
the predicate-counting abstractions of $w$ and $w'$ are the same
and $w'$ uses $O((n + |\nonterm|) \log (n + |\nonterm|))$ different letters.
Notice that this is an almost linear bound.
In the case when $\Aut$ is a parametric symbolic grammar, the 
size of the alphabet is $O((n + |\nonterm|) \log \ell (n + |\nonterm|)) $, where $\ell$
represents the maximum length of the right-hand side of a production in $\Aut$.
Furthermore, we show that we can compute an existential 
$\theory$+LIA formula $\varphi_{\Aut}$ that captures this predicate-counting 
abstractions. The formula $\varphi_{\Aut}$ can be solved easily in the standard
SMT framework of DPLL($\theory$, LIA) (e.g.\ \cite{KS08}), which uses LIA and
$\theory$ solvers separately to add blocking lemmas. It follows immediately that
we obtain decision procedures for analyzing satisfiability of 
predicate-counting abstractions (possibly restricted with additional LIA
formulas) with a tight complexity upper bound: if $T$ is \np-complete (resp.
\pspace-complete), then our problem is also \np-complete (resp.
\pspace-complete). 

Since string constraints are defined over the Unicode alphabet\footnote{See 
the SMT-LIB 2.6 specification 
\url{https://smtlib.cs.uiowa.edu/theories-UnicodeStrings.shtml}},
one natural application of our result is in checking unsatisfiability of string
constraints. By means of predicate-counting abstractions, we show how string 
constraints can be abstracted into the Parikh image of a symbolic grammar
with an additional LIA restriction. Here, we allow an expressive and well-known
subclass of
string constraints (in particular, commonly used subclass of QF\_SLIA theory of 
SMT-LIB 2.6 \cite{SMTLIB2}), which permits string
concatenation, string equations, replace, regular constraints, contains, 
prefix-of, and suffix-of. Note that unsatisfiability of the latter
implies unsatisfiability of the original string constraint, but not the
converse;
we are not aware of any non-trivial class for which the converse implication would hold,
a trivial one is over a unary alphabet.
At the same time our result admits an easy extension to sequence theories,
which permit general effective boolean algebras (e.g.\ see \cite{JLMR23}).

We have implemented this translation, which
takes an SMT-LIB file and produces a quantifier-free LIA formula, which can be
easily checked using SMT-solvers. Our experimental results show that our
procedure can substantially outperform existing string solvers for proving
unsatisfiability (details are in Figure \ref{fig:results}). 

Finally, as mentioned above, our paper establishes Parikh's Theorem for 
generalizations of symbolic automata, i.e., parametric (symbolic) grammars and 
parametric (symbolic) pushdown automata. Such formalisms are highly expressive,
e.g., can express Dyck languages with \emph{infinitely} many parenthesis 
symbols. This has 
many potential applications. The first application is the support of symbolic
context-free constraints, 
i.e., an expression of the form $x \in L$, where $L$ is given by a parametric
grammar. In fact, classical results on string analysis (e.g. 
\cite{CMS03,Min05}) heavily use context-free constraints, which are not
supported by SMT-LIB 2.6, but are supported by a handful of modern string 
solvers (e.g. TRAU \cite{TRAU}). The second application is a partial support 
of a ``future-looking'' feature in SMT-LIB 2.6: \texttt{to\_re}, which converts a
string (possibly with variable names) to regular expressions. 
This is highly
expressive, e.g., it allows encoding word equations with Kleene star
like $xy = z^*$. 
Existing benchmarks allow only a very limited usage of \texttt{to\_re}: strings 
with only constants (i.e. no variables) as input. Using parametric grammars, we 
can encode some interesting use cases of \texttt{to\_re}
beyond only string with constants, e.g., we can encode parametric regular 
constraints of the form $x \in y^*$, where both $x$ and $y$ are variables. 
Finally, parametric pushdown automata strictly extend the model of symbolic 
visibly pushdown automata \cite{DA14}, which has applications to dynamic 
analysis of programs. By allowing parameters and pushing values onto stack, 
our model allows some support of static analysis as well (of course with 
restrictions, for otherwise decidability would result).

\OMIT{
Finally, we mention also further applications of our results, including the 
validation of XML documents, LTL model checking, and decision procedures over 
sequences. 
}


\paragraph{Organization.}
We fix our notation and introduce our notion of parametric context-free grammar in Section
\ref{sec:model}. Also in this section, we show that these grammars admit a representation as parametric pushdown automata. In Section~\ref{sec:parikh} we prove our new
Parikh's Theorem for parametric symbolic grammars. In Section~\ref{sec:string}, we
provide an abstraction of string constraint solving
via predicate-counting constraints, and outline an extension to sequence
theories. We describe our implementation and report
our experimental results in Section~\ref{sec:experiments}.
We conclude the paper in Section~\ref{sec:conc} with related work and future work.





\section{Models}
\label{sec:model}

We start by introducing some basic notation for quantifier-free theories.
Then we introduce parametric context-free grammars and finally an equivalent model of pushdown systems.

\subsection{Preliminaries}

For simplicity, we follow a model-theoretic approach to define our symbolic
alphabets. This allows a more convenient treatment of \emph{parameters}
in our automata (e.g. see \cite{JLMR23}).
Let $\vocab$ be a set of vocabulary symbols. We fix a
$\vocab$-structure $\struct = (D; I)$, where $D$ can be a finite or an infinite
set (i.e.\ the universe) and $I$
maps each function/relation symbol in $\vocab$ to a function/relation over
$D$. The elements of our sequences will range over $D$.
We assume that the quantifier-free theory $T_{\struct}$ over $\struct$ (including
equality) is decidable. Examples of such $T_{\struct}$ are abound from
SMT, e.g., Linear Real Arithmetic and Linear Integer Arithmetic.
We write $T$ instead of $T_{\struct}$, when $\struct$ is clear.
Our quantifier-free formula will use \emph{uninterpreted
$T$-constants} (which will represent the parameters)
$a,b,c,\ldots$, and local variables $x,y,z,\ldots$.
We use $\Constants$ to denote the set of all uninterpreted $T$-constants and $\Variables$ to denote the set of all local variables.

We write $\phi(\params, \locals)$ for a formula that is a Boolean combination of terms constructed from functions/relations in $\vocab$, uninterpreted constants $\params \subseteq \Constants$, and local variables $\locals \subseteq \Variables$.
We say such a formula is a formula of $T$.
An existential formula of $T$ is of the form $\exists x_1, \ldots, x_n\ .\ \phi$,
where $\phi$ is a quantifier-free formula of $T$.
An \emph{interpretation} or \emph{assignment} of the constants (resp.\
variables) is a map
$\Constants \rightarrow D$ (resp.\ $\Variables \rightarrow D$).
We write $\phi(\iparam, \ivar)$ for the formula under these interpretations.
We write $T \models \phi(\iparam, \ivar)$ if $\phi$ is true in $\struct$ under interpretations $\iparam, \ivar$.
A formula $\varphi$ is satisfiable if there are interpretations $\iparam, \ivar$ such that the formula becomes true in $\struct$.

We write $T_1 + T_2$ for the Boolean combination of theories $T_1$ and $T_2$.
That is, quantifier-free Boolean combinations of formulas $\phi$ where $\phi$ is a quantifier-free formula either of $T_1$ or of $T_2$.
When $T_1$ and $T_2$ satisfy certain conditions, decision procedures for $T_1$ and $T_2$ can be combined, e.g.\ with Nelson-Oppen~\cite{NO79}.


\subsection{Parametric Context-Free Grammars}

We introduce the notion of \defn{parametric context-free grammars},
which generalize context-free grammars to infinite alphabets
in a similar way as parametric automata~\cite{FK20,FJL22,FL22}
generalize finite automata;
note that parametric automata in turn
generalize both symbolic automata \cite{DV21,symbolic-power} and
variable automata \cite{GKS10} and are
expressible incomparable 
 to symbolic register automata \cite{SRA} (i.e. neither
subsumes the other).

In brief, a production $(A, \alpha, \phi)$ of a parametric context-free grammar is an extension of a production $(A, \alpha)$ of a classical context-free grammar.
The production replaces a non-terminal $A$.
However, $\alpha$ is not a sequence of characters and non-terminals, but a sequence of local variables and non-terminals, e.g.\ $y A y$.
The final component $\phi$ is a guard over both the parameters of the grammar and the local variables.
E.g. $\phi$ may assert $y = x$ for the parameter $x$.
Hence, if $x$ took the concrete value $a$, the production would replace $A$ with $a A a$.

\begin{definition}[Parametric Grammar]
    A \defn{parametric grammar} over alphabet theory $T$ is
    of the form $\Grammar = (\params, \nonterm, D, P, S)$,
    where $\params$ is a finite set of parameters,
    $\nonterm$ are the nonterminal symbols,
    $D$ is the (perhaps infinite) set of symbols from the domain,
    $S$ is a starting symbol
    and $P$ is a finite set of triples $(A,\alpha,\phi)$ from $\nonterm \times
    (\locals \cup \nonterm)^* \times T(\params, \locals)$, where $\locals$ is the
    set of local variables.
\end{definition}
Here, \emph{parameters}
are uninterpreted $T$-constants, i.e.\
$\params \subseteq \Constants$;
the $\locals$ variables are ``local'' in the sense that they are instantiated each time the production is used.
Formulas that appear in the productions from $P$ will be referred to as
\emph{guards}, since they restrict which symbols can be produced.

The semantics is defined by generalizing the rewriting-style definition of derivation for context-free grammars.
A derivation begins with the singleton sequence $\beta_0$ consisting only of the starting symbol $S$.
From a sequence $\beta A \beta'$ we can derive $\beta \alpha' \beta'$,
when $\alpha'$ can be derived from $A$, which is defined as follows:
For interpretations $\iparam$ of $\params$ and $\ivar$ of $\locals$ we write
$\alpha[\params/\iparam,\locals/\ivar]$
for the substitution of the constants in $\params$ by their interpretation under $\iparam$ and the variables $\locals$ by their interpretation under $\ivar$.
The nonterminal
$A$ can be rewritten by a rule $(A, \alpha, \varphi(\params, \locals))$
with $\alpha' = \alpha[\params/\iparam,\locals/\ivar]$ when $T \models \varphi(\iparam, \ivar)$.
We call $\alpha[\params/\iparam,\locals/\ivar]$ an \emph{instantiation} of $\alpha$ and often denote it by $\alpha'$.
We will often refer to a rule $A \to \alpha$ in this case, suppressing the actual instantiation.

Then $L_\iparam(\Grammar)$ is the set of nonterminal-free sequences from $D^*$ that can be derived by $\Grammar$ for an interpretation $\iparam$ of the parameters
and $L(\Grammar) = \bigcup_\iparam L_\iparam(\Grammar)$.

\begin{example}
Consider a grammar $\Grammar$ with one nonterminal, no parameters and rules
$$
(S, ySy, \top), \quad (S, yy, \top)
$$
where the formula $\top$ is always true. Then the generated language is
$
L(\Grammar) = \{w w^R \: : \: w \in D^+\} \ .
$
Let us add a single parameter $x$ and change the rules to:
$$
(S, yxSxy, \top), \quad (S, yxxy, \top) \ .
$$
Then the language is
$$
    L(\Grammar) = \bigcup_{c\in D} \{d_1 c d_2 c \cdots d_k c c d_k c \cdots c d_1 \: : \: k \in \N_+, d_1, \ldots, d_k \in D\} \ .
$$
Notice that the parameter $x$ takes the value $c$ in all productions, but the local variable $y$ can take a different value ($d_\idxi$) during each application of the productions.

Assuming $D = \N$ we can add a condition on the production:
$$
(S, yxSxy', \text{``}y > x \land y + y' = 0\text{''}), \quad (S, yxxy', \text{``}y > x \land y + y' = 0\text{''})
$$
which,
using $(d_\idxi)$ to denote the negative number $-d_\idxi$,
results in a language
$$
\bigcup_{c \in D} \{d_1 c d_2 c \cdots d_k c c (d_k) c \cdots c (d_1) \: : \: k \in \N_+, d_1, \ldots, d_k > c \} \enspace .
$$
\end{example}

In proofs dealing the properties of the derivations (and not the derived sequences),
we focus on rules $A \to \alpha$ and their instantiations $\alpha'$,
and not the guards, which are defined implicitly by the choice of rule.

Note that when all rules of the grammar are of the form $A \to yB$ with $|y| = 1$ and $B \in \nonterm$ or $A \to \epsilon$,
then we can think of $\Grammar$ as an automaton (with $\nonterm$ taking the
role of states, $S$ being the initial state and states such that $B \to
\epsilon$ being the final states). Such an automaton is an extension of both
symbolic automata \cite{DV21,symbolic-power,veanes12} and variable automata
\cite{GKS10,FK20}, which is referred to as parametric (symbolic) automaton
\cite{FL22,FJL22,JLMR23}.



\begin{proposition}
    Assume that $T$ is solvable in \np (resp.\ \pspace). Then, deciding
    nonemptiness of a parametric grammar over $T$ is in \np (resp.\ \pspace).
\end{proposition}
\begin{proof}
The proof is a generalization of a standard proof for context-free grammars
and proofs for parametric automata, see~\cite{FK20,FJL22,FL22}.
The language $L(\Grammar)$ is nonempty if and only if it is nonempty for some parameters $\iparam$, which we will existentially quantify over.

As in the case of standard context-free grammars,
(for fixed parameters) the $L_\iparam(\Grammar)$ is nonempty
if and only if we can order a subset of nonterminals
$A_0, A_1, \ldots, A_k$, where $S = A_0$,
and for each find a rule $(A_i,\alpha_i,\phi_i)$ such that:
if $A_j$ is in $\alpha_i$ then $j > i$;
and for each $i$ the guard $\phi_i(\iparam,\ivar)$ is satisfiable (for some $\ivar$).
Hence we guess the sequence $A_0, A_1, \ldots, A_k$,
verify the first condition
and then verify the formula
$
\exists\iparam\ .\ \bigwedge_{i=0}^k \exists\ivar\ .\ \phi_i(\iparam, \ivar)
$.
Clearly, if $T$ is in \np (resp.\ \pspace) then the above algorithm is in the same class.
\end{proof}

\subsection{Parametric pushdown automata}
\label{sec:parametric-pushdown-automata}

As in the case of usual context-free grammars, there is a pushdown equivalent of our parametric grammar.
Some care is needed, as some combinations of allowed transitions lead to a much more powerful model:
for instance, it is not difficult to show that if we do not allow parameters then our grammar model cannot express the language $\bigcup_{d \in D} d^*$,
yet it is easy to come up with a pushdown automaton which can recognize such languages (with no parameters): it is enough to store the first symbol on the stack
and then compare each consecutive symbol with it; hence such a run should not be allowed.
However, it seems natural that we should allow storing elements of $D$
on the stack, so that, say, palindromes can be recognized.
As a solution, if $\Aut$ sees an element from $D$ on the top of the stack
then it must pop it from the stack.
It cannot push and cannot make an $\epsilon$-transition.

Nondeterministic parametric pushdown automata are formally defined as follows (cf.~standard definition~\cite{sipser13}).
We explain some restrictions on the definition below.

\begin{definition}[Nondeterministic Parametric Pushdown Automata]
A nondeterministic parametric pushdown automaton is a tuple $(\controls, D, \Gamma, \transrel, q_0, F)$,
where as usual $\controls$ is a finite set of states,
$\Gamma$ is finite stack alphabet, we require that $\Gamma \cap D = \emptyset$,
$q_0 \in Q$ is the starting state, $F \subseteq \controls$ is a set of accepting states
and $\transrel = \transrel_{\Gamma \cup \{\epsilon\}} \cup \transrel_{D}
\cup \transrel_{\epsilon, \Gamma \cup \{\epsilon\}} \cup \transrel_{\epsilon,D}$ is a transition relation:
\begin{align*}
\transrel_{\Gamma \cup \{\epsilon\}}
	&\subseteq
\underbrace{\controls}_{\text{state}} \times
    \underbrace{(\Gamma \cup \epsilon)}_{\text{stack top symbol}} \times
    \underbrace{T(\curr,\locals, \params)}_{\text{guard}} \times
    \underbrace{\controls}_{\text{new state}} \times
    \underbrace{(\Gamma \cup \params \cup \locals \cup \curr)^*}_{\text{pushed to stack}}\\
\transrel_{D}
	&\subseteq
\underbrace{\controls}_{\text{state}} \times
    \underbrace{T(\curr,\stktop, \locals, \params)}_{\text{guard}} \times
    \underbrace{\controls}_{\text{new state}}  \\
\transrel_{\epsilon, \Gamma \cup \{\epsilon\}}
	&\subseteq
\underbrace{\controls}_{\text{state}} \times
    \underbrace{(\Gamma \cup \epsilon)}_{\text{stack top symbol}} \times
    \underbrace{T(\locals, \params)}_{\text{guard}} \times
    \underbrace{\controls}_{\text{new state}} \times
    \underbrace{(\Gamma \cup \params \cup \locals)^*}_{\text{pushed to stack}}\\
\transrel_{\epsilon, D}
	&\subseteq
\underbrace{\controls}_{\text{state}} \times
    \underbrace{T(\stktop, \locals, \params)}_{\text{guard}} \times
    \underbrace{\controls}_{\text{new state}}       
    \ .
\end{align*}
\end{definition}

Our definition is a strict extension of symbolic visibly pushdown automata 
\cite{DA14}.
In the above definition $\curr$ is a variable bound to the symbol from $D$ read by the automaton,
$\stktop$ is a variable bound to stack top symbol, which is from $D$.
The four cases of transition function are for ease of presentation,
as in principle one could define $\Delta$ as one set with some syntactic conditions.
The case-distinction is as follows:
In $\transrel_{\Gamma \cup \{\epsilon\}} \cup \transrel_{D}$
the automaton reads an input letter (bound to variable $\curr$) 
and in $\transrel_{\epsilon, \Gamma \cup \{\epsilon\}} \cup \transrel_{\epsilon,D}$ it does not, i.e.\ those are $\epsilon$-transitions
and they do not refer to $\curr$ in the guards, nor in the word pushed to stack.
In $\transrel_{\Gamma \cup \{\epsilon\}} \cup \transrel_{\epsilon, \Gamma \cup \{\epsilon\}}$ the stack topmost symbol is from $\Gamma$ or we do not read the stack at all;
in $\transrel_{D} \cup \transrel_{\epsilon, D}$ the stack top-most symbol is
from $D$, in which case we are not allowed to push anything to the stack;
on the other hand we can use it in the guard, say for comparison with $\curr$.
The reason for not allowing pushing to the stack in this case is that
we do not want to copy the stack contents,
which easily leads to recognition of language $\cup_{d \in D} d^*$,
which should not be recognized without parameters.

Like the case of parametric grammars,
the sequence pushed to the stack may in general depend on the read character $\curr$, some local variables $\locals$ and the parameter interpretation $\iparam$.
We require that when $\alpha'$ is actually pushed to stack
(say for a transition in $\transrel_{\Gamma \cup \epsilon}$)
then $\alpha' = \alpha[\curr/d, \params/\iparam, \locals/\ivar]$, where $d$ is the character read and $\ivar$ is any assignment to $\locals$ such that
$\phi(d, \cha, \iparam, \ivar)$ holds, where $\phi$ is the guard of the rule and $\cha$ is the top of stack character;
$\alpha'$ for $\epsilon$-transitions is defined similarly, i.e.\
as $\alpha[\params/\iparam, \locals/\ivar]$.

Note, the guards can provide expressive power:
E.g.\ for palindromes, we can store the first half of the read word on the stack and then, for the second half, check equality with the read symbol while popping character by character from the stack.
That is, using the guard $\stktop = \curr$.

%
%
Let us describe the semantics,
we will focus on $\transrel_{\Gamma \cup \{\epsilon\}}$,
the other cases are defined similarly.
A configuration is a tuple $(q, \word)$ where $q \in \controls$
is a state and $\word \in (\Gamma \cup D)^\ast$ is the stack contents;
as a convention, we assume that the stack top-most symbol is the first in $\word$.
Take a~configuration $(q, \cha \word)$, where $\cha \in \Gamma \cup \epsilon$, and a transition
$(q, \cha, \phi, q', \alpha)$,
let $d$ be the symbol read by the automaton.
If for some $\ivar$ the $\phi(d,\iparam, \ivar)$ holds 
then the automaton can change the configuration to
$(q', \alpha[\curr/d, \params/\iparam, \locals/\ivar]w)$.
Note that if $a \in \Gamma$ then we need to pop it from the stack,
and if $a = \epsilon$ then the transition does not depend on the stack contents.
For $\transrel_{D}$ the move is defined analogously,
but when the stack contents is $s w$ for $s \in D$,
the guard is evaluated as $\phi(d, s, \iparam, \ivar)$.
The semantics of $\epsilon$-transitions is defined similarly.

A word $\word \in D^\ast$ (for parameter interpretation $\iparam$) is accepted if there is a run for $\word$
from $(q_0, \epsilon)$ to $(q, \word')$ for some $\word'$ and $q \in F$.
By $L_\iparam(\Aut)$ we denote the language recognized by a parametric NPDA $\Aut$
for a given interpretation $\iparam$ of parameters,
and define $L(\Aut) = \bigcup_\iparam L_\iparam(\Aut)$.

\begin{theorem}
The class of languages recognized by parametric context-free grammars
and parametric non-deterministic pushdown automata coincide.
\end{theorem}

The equality is shown using two natural inclusions,
proven in the Lemmata below.

\begin{lemma}
Given a parametric grammar $\Grammar$ we can compute in polynomial-time a
    parametric NPDA $\Aut$
of size linear in the size of $\Grammar$
such that for each parameter interpretation $\iparam$ we have
$L_\iparam(\Grammar) = L_\iparam(\Aut)$.
\end{lemma}
\begin{proof}
The proof is an adaptation of the classic proof, see~\cite[Lemma 2.21]{sipser13};
note that since $\epsilon$-transitions are allowed, we do not need Greibach normal form,
which is a little cumbersome in parametric setting.

Given a sequence $w$ the $\Aut$ will simulate the derivation of $\Grammar$
by always greedily expanding the left-most nonterminal and matching the left-most unmatched letter of the input sequence.
We use the same parameter interpretation $\iparam$ as $\Grammar$ does.
The automaton has three states $q_0, q$, and $q_f$.
The starting state is $q_0$ and $q_f$ is the unique accepting state.
The $\Gamma$ is $\nonterm \cup \{\bot\}$,
where $\bot$ represents the stack bottom.
In $q_0$, the automaton $\Aut$ pushes $S \bot$ to the stack (so $S$ is top-most) and moves to $q$, here $S$ is the starting symbol of the grammar..
In $q$ if $\Aut$ sees $\bot$ on the stack then it moves to $q_f$ and accepts (it cannot proceed).
Otherwise, if the topmost symbol is $A \in \nonterm$ then
$\Aut$ chooses (nondeterministically) a rule $A \to \alpha$ and its (valid) instantiation $\alpha'$,
pops $A$ and pushes $\alpha'$ to the stack.
If the topmost symbol is $d \in D$ and the next symbol is $d$ (that is $\stktop = \curr$)
then $\Aut$ pops the letter and reads the next symbol from the input
(and stays in $q$). It is easy to see that the resulting automaton $\Aut$ has
size linear in $\Grammar$, and can furthermore be computed in polynomial-time.
A slight modification of a standard proof shows the for each instantions of parameters  $\iparam$
we have $L_\iparam(\Grammar) = L_\iparam(\Aut)$,
so also $L(\Grammar) = L(\Aut)$, as claimed.
\end{proof}

The other direction is slightly more involved.
Our construction is exponential in the maximum length of $\alpha$ appearing in a transition
$(q, c, \phi, q', \alpha) \in \transrel_{\Gamma \cup \epsilon}$
or in a transition $(q, \phi, q', \alpha) \in \transrel_{\epsilon, \Gamma \cup \epsilon}$.
Unlike the non-parametric case, we cannot split pushing transitions
so that each transition pushes at most one symbol to the stack.
This is because the values of $\locals$ and $\curr$ cannot be transferred across separate transitions.
Thus a rule pushing $\alpha$ to the stack
needs an exponential number of productions to handle all
sequences of intermediate states that may occur while $\alpha$ is later being
popped. However, so long as we fix the maximum length of such $\alpha$, the
resulting grammar is of polynomial size and can be computed in
polynomial-time.

\begin{lemma}
Given a parametric NPDA $\Aut$ there is a parametric grammar $\Grammar$
such that for each parameter interpretation $\iparam$ we have
$L_\iparam(\Grammar) = L_\iparam(\Aut)$.
The size of $\Grammar$ is exponential in the maximum number $M$ of push
    symbols appearing in any transition of $\Aut$. When $M$ is fixed, the
    algorithm runs in polynomial-time.
\end{lemma}
\begin{proof}
We modify a standard construction cf.~\cite[Lemma 2.27]{sipser13}.
For a fixed parameter interpretation $\iparam$,
$L_\iparam(A_{q,q'})$ is the language recognized by $\Aut$ with starting state $q$,
final state $q'$.
We modify $\Aut$ so that:
\begin{itemize}
\item it has a single accepting state
\item it empties the stack before accepting
\item in each move it either pushes something (perhaps the empty sequence)
to the stack or pops from the stack, but not both.
\end{itemize}
The first two conditions are easy to ensure, the last depends on the form of the transition:
\begin{itemize}
\item If the transition is from $\transrel_D \cup \transrel_{\epsilon, D}$
then it does not push, as required.
\item If the transition is from $\transrel_{\Gamma \cup \epsilon} \cup \transrel_{\epsilon, \Gamma \cup \epsilon}$
and it reads $\epsilon$ from the stack,
then it does not pop from the stack, as required.
\item If the transition is
from $\transrel_{\Gamma \cup \epsilon} \cup \transrel_{\epsilon, \Gamma \cup \epsilon}$
for a topmost symbol $\gamma \in \Gamma$
then we create a new state $q_{q, \gamma}$ and create an $\epsilon$-transition
from $q$ that removes $\gamma$ without reading a letter and
goes to $q_{q,\gamma}$.
Then from $q_{q,\gamma}$ the automaton ignores the stack and acts as if it were in $q$ with $\gamma$ on top of the stack.
\end{itemize}
The defined automaton recognizes the same language
(for parameter interpretation $\iparam$.)

Note, that the conditions on the transition relation
when the topmost symbol is from $D$
are tailored so that the above separation of popping and pushing is possible.

In a standard proof of equivalence of NPDAs and CFGs,
cf~\cite[Lemma 2.27]{sipser13},
the computation of $\Aut$ is split into parts
in which it empties the stack (from the symbols it introduced).
The assumption that $\Aut$ can push at most one element to the stack
makes the proof easier; however, we need to push more symbols.
But this only means that when $a_1 a_2 \cdots a_k$
is pushed to the stack, the computation is split into $k$ subcomputations,
in which it removes $a_1, a_2, \ldots, a_k$ from the stack.

To be more precise,
let $\Gamma = \nonterm \times \controls \times \controls$
and denote its elements by $A_{q,q'}$,
with the intention that $L_\iparam(A_{q,q'})$ is the language of words
such that $\Aut$ starting in $q$ (and empty stack) will go on this word to the empty stack and state $q'$.
In particular,
we will set $A_{q_0,q_f}$ as the starting symbol,
where $q_0$ is the starting state and $q_f$ the unique final state,
and then $L(A_{q_0,q_f}) = L(\Aut)$.

Recall the classic construction, in which case each rule pushing
an element to the stack pushes at most one symbol.
When describing the computation taking the automaton from $q$ to $q'$
and from the empty stack to the empty stack,
i.e.\ corresponding to the nonterminal $A_{q,q'}$,
either the stack is emptied somewhere on the way, say at state $q''$,
which means that we have a rule $A_{q,q'} \to A_{q,q''}A_{q'',q'}$,
or it is emptied at the last step. Hence if the first transition
pushes $s$ to the stack, the last removes it from the stack
and in the meantime the computation is as if it started and ended 
on an empty stack.
Hence the rule is of the form $A_{q,q'} \to a A_{p,p'} b$
such that there is a transition from $q$ to $p$, reading $a$, pushing $s$
and a transition from $p'$ to $q'$ popping $s$, reading $b$
(both $a,b$ can be a letter or $\epsilon$).

In our case we cannot assume that a transition pushes just a single symbol to the stack,
as the sequence pushed may contain the same local variable or $\curr$ several times
and splitting into many rules would lose this connection.
Hence we need to consider rules pushing, say,
$s_1\cdots s_k$ and the $k$ rules that pop those letters
(and in between $\Aut$ acts as if it were starting and ending on the empty stack),
at the same time and ``compound'' their computation.

We include the rules given below;
The first bullet point states that a non-terminal $A_{q,q}$ can be rewritten to $\epsilon$.
That simulates a move from $q$ to $q$ without firing any transitions.
In the second bullet, the production
$A_{q,q''} \to A_{q,q'}A_{q',q''}$
handles the case where the stack is emptied at state $q'$ on the run from $q$ to $q''$.
\begin{itemize}
\item $A_{q,q} \to \epsilon$
\item for each $q,q',q'' \in \controls$ a rule
$$
A_{q,q''} \to A_{q,q'}A_{q',q''}
$$
\item for each $(q_0,\epsilon, \phi_0, q_1, s_1\cdots s_k) \in \transrel_{\Gamma \cup \epsilon} \cup \transrel_{\epsilon, \Gamma \cup \epsilon}$
(so the one pushing $s_1\cdots s_k$ to the stack and not looking at the stack)
such that there are $q_2, \ldots, q_{k+1}$
(states in which $s_1, \ldots, s_k$ are popped)
such that for all $2 \leq i \leq k$
we have
$(q_i,\phi_i, q_{i+1}) \in \transrel_D \cup \transrel_{\epsilon, D}$ or 
$(q_i,s_i', \phi_i, q_{i+1},\epsilon) \in \transrel_{\Gamma \cup \epsilon} \cup \transrel_{\epsilon, \Gamma \cup \epsilon}$
(the transitions popping those letters)
we add a production 
$$
A_{q_0,q_{k+1}} \to y_0 A_{q_1,q_2} y_1 A_{q_2,q_3} y_2 \cdots A_{q_k,q_{k+1}} y_k
$$
and a guard
$$
    \phi_0[\curr/y_0, \params, \locals/\locals_0] \land
    \bigwedge_{i=1}^k \phi_i[\curr/y_i, \stktop/s_i, \params, \locals/\locals_\idxi]
$$
the $y_i$ is the symbol read by the $i$-th transition,
$\locals_i$ are the local variables of the $i$-th transition,
$s'_i$ is the symbol popped by $i$-th transition
and $\phi_i$ is the guard of the $i$-th transition;
those are described in detail below.
Note that if the $i$-th transition for $0 \leq i \leq k$
is an $\epsilon$-transition then $\phi_i$ does no depend on the $\curr$,
and if it does not read the stack top element, then $\phi_i$ does not use $\stktop$,
but we write like this to streamline the argument.

Concerning $y_i$, corresponding to letters read by popping transitions,
if the $i$-th transition is an $\epsilon$ transition
then $y_i = \epsilon$ and otherwise it is a fresh local variable
(which appears in the guard).
Concerning $s'_i$, they ``should be'' $s_i$,
but the problem arises when $s_i$ is equal to $\curr$,
i.e.\ $s_i$ pushed to the stack is the read letter;
then $\curr$ ``should be'' $y_0$.
Hence, if $s_i = \curr$ we define $s'_i = y_0$.
That is, at $q_0$ the read symbol $y_0$ was pushed to the stack.
Otherwise --- if $s_i \neq \curr$ --- we set $s'_i = s_i$,
which could be both an element from $\Gamma$ or $D$.

We have different fresh copies $\locals_\idxi$ of the local variables.
This is because multiple pushdown transitions are encoded in a single rule.
The value of the local variables may differ for each transition fired.
Hence, we need separate copies for the different guards $\phi_i$ combined in the grammar production.
Notice that these separate copies are never pushed onto the stack as there is only one pushing transition represented by the production.
\end{itemize}

The proof is now a generalization of the standard one~\cite[Lemma 2.27]{sipser13}.
When starting from $q$ with an empty stack and ending in $q'$,
then if we reach the empty stack somewhere on the way we use $A_{q,q''} \to A_{q,q'}A_{q',q''}$.
Otherwise, we push some $s_1\cdots s_k$ to the stack
and take them from the stack one by one: i.e.\ for each $i$ there is the first moment
when $s_i$ is taken from the stack and from the moment we took $s_{i-1}$
right before we take $s_i$ the $\Aut$ acts as if on empty stack.
\end{proof}


\section{Computing Parikh images} \label{sec:parikh}

We first generalise the notion of a Parikh image to the parametric setting.
Then we discuss the construction of Parikh images, and the complexity of a related decision problem.

\subsection{Definition}

We first recall the classical definition of the Parikh image of a language.
Take a finite alphabet $\Sigma = \{\cha_1, \ldots, \cha_n\}$.
For a word $\word \in \Sigma^\ast$, let
$\countof{\cha}{\word}$
be the number of occurrences of the character $\cha$ in the word $\word$.
For a linearisation $\cha_1, \ldots, \cha_n$ of the characters of $\Sigma$, the Parikh image $P(\word)$ of $\word$ is a mapping
$f : \{1, \ldots, n\} \to \N$
such that $f(\idxi) = \countof{\cha_\idxi}{\word}$ for all $1 \leq \idxi \leq n$.
For a language $L \subseteq \Sigma^\ast$ the Parikh image is the set of mappings
$\{P(\word) : \word \in L\}$.
That is, the Parikh image counts the number of occurrences of each character in each word of $L$.

Parametric context-free grammars may have large or even infinite alphabets.
Counting each of the characters is either impractical or impossible.
Hence, we define a version of the Parikh image that is relative to a sequence of predicates.
We then count the number of characters satisfying each predicate, rather than the number of each individual character.

\begin{definition}[Parametric Parikh Image]
    For a word $\word = \cha_1 \ldots \cha_k \in D^\ast$ and $T$-formula $\psi(\curr)$ over one local variable $\curr$, the count
    $\countof{\psi}{\word}$
    is the number of positions $\idxi$ of $\word$ such that
    $T \models \psi(\cha_\idxi)$.
    For a sequence $\Psi := \psi_1,\ldots,\psi_n$ of $T$-formulas, the Parikh image $P_\Psi(\word)$ of $\word$ over $\Psi$ is a mapping
    $$f : \{1, \ldots, n\} \to \N$$
    with $f(\idxi) = \countof{\psi_\idxi}{\word}$ for all $\idxi$.
    For a parametric grammar $\Grammar$ over $T$, the Parikh image of $\Grammar$ over $\Psi$ is
    \[
        P_{\Psi}(L(\Grammar)) \defeq
        \{P_{\Psi}(\word) : \word \in L\} \ .
    \]
\end{definition}

Let us fix a parametric grammar $\Grammar$ over $T$ and a sequence $\Psi := \psi_1,\ldots,\psi_n$ of $T$-formulas.

\subsection{Representing the Parikh Image}

We can construct a formula in the combined theory of $T$ and quantifier-free Presburger arithmetic that represents the Parikh image of a given parametric grammar $\Grammar$.
Such a formula is polynomial in size and can be used as part of a query to an SMT solver to solve decision problems over the grammar.
Below, let $\mathrm{QFPresburger}$ be the theory of quantifier-free fragment 
of Presburger arithmetic.

\begin{theorem}
\label{thm:Parikh_image_represented}
Let $\Grammar$ be a parametric grammar over the theory $T$ and take a sequence $\Psi := \psi_1,\ldots,\psi_n$ of $T$-formulas.
There is an existential $T + \mathrm{QFPresburger}$ formula $\phi(x_1, \ldots, x_n)$ of size polynomial in the size of $\Grammar$ such that
$f \in P_\Psi(L(\Grammar))$ iff $\phi(f(1), \ldots, f(n))$ holds.
\end{theorem}

In principle there are exponentially many (in $n$)
different symbols $\cha$ which yield different images $P_\Psi(\cha)$.
A naive approach to computing Parikh images would compute the possible images $f$ of some character $\cha$ and calculate their possible sums in words generated by $\Grammar$.
Since this may require an exponential number of different characters, the naive approach implies at least an exponential running time.
We show that this is not the case and that the formula can be polynomial in size.

The main step of the proof is to show that if $f \in P_\Psi(L(\Grammar))$
then there is a sub-domain $D_0 \subset D$ of nearly linear (in $n$ and $|N|$) size such that
$f \in P_\Psi(L(\Grammar) \cap D_0^*)$
(the actual statement is more precise).
This essentially reduces the problem to the standard case of a finite alphabet;
some details are still needed (like finding the value of the parameters, finding the exact subdomain, etc.), but this is the crucial step.
Note that this finitization method does not hold for some models
over infinite alphabets, e.g., semilinear data automata of \cite{FL22}, which
can impose that there are exponentially many different elements of the domain in
a given predicate $\pred$. A similar result was shown independently to obtain
Carath\'eodory bounds for integer cones~\cite{ES06} and used to derive 
parallel results on the complexity of non-emptiness of symbolic tree
automata~\cite{R23}.

To show the existence of such a sub-domain,
we consider derivations of a word and we focus on the number
of times (counts) each production is used.
Note that the Parikh image is determined by those counts.
Hence, the counts are an ``intermediate notion'' between the exact derivations and Parikh images and we focus on them.
For classic context-free grammars it is well-known that a derivation with given counts exists
if and only if those counts satisfy simple arithmetic constraints.
Those constraints can be formulated in terms of sums of mappings,
which are similar to the Parikh image mappings.

Then, using combinatorial arguments, we show that if we have a large number of such mappings (we use many productions)
then there are subsets that have the same effect: i.e.\ in any derivation we can replace a subset of the productions with another subset without changing the Parikh image of the result.
Moreover, we can construct weights for such sets and ensure that each time we make a replacement, the weight drops,
which means that the replacing terminates at some point.
The result has to be a derivation with a small number of different productions used.

We illustrate this process with a small example over characters rather than productions.
Take the language $D^\ast$.
Assume we have three predicates $\Psi = \psi_1, \psi_2, \psi_3$.
Take a word $a_1 a_2 a_3 a_4 a_5$ using five different characters and suppose the Parikh image (using vectors to represent the maps) is
\[
    P_\Psi(a_1 a_2 a_3 a_4 a_5) =
        (1, 0, 0) +
        (1, 1, 1) +
        (1, 1, 0) +
        (1, 0, 1) +
        (0, 0, 0) \ .
\]
Notice
\[
    P_\Psi(a_1 a_2) =
        (1, 0, 0) +
        (1, 1, 1) =
    P_\Psi(a_3 a_4 a_5) =
        (1, 1, 0) +
        (1, 0, 1) +
        (0, 0, 0)
\]
and hence we can construct the same Parikh image using only $a_1$ and $a_2$.
That is $P_\Psi(a_1 a_2 a_1 a_2) = P_\Psi(a_1 a_2 a_3 a_4 a_5)$.
Our proof shows that if a large number of different productions are used, the same principle will allow us to find subsets with the same sum.

To make the above intuition formal,
we first recall that there is a derivation of a CFG with a given number of times each production is used,
if and only if those counts satisfy a simple arithmetical relation~\cite{VermaSS05};
informally: for each nonterminal, the number of times it is introduced and expanded are the same (except for starting nonterminal)
plus a condition guaranteeing a variant of being connected;
this characterization is similar in spirit and proof to the Euler condition for directed graphs.
Formally, for a rule $A \to \alpha$ let $n_{A, \alpha}$ be its count,
then
\begin{lemma}[cf.~{\cite[Thm.~3, 4]{VermaSS05}}]
\label{lem:derivation_conditions_simple}
There is a derivation of a CFG (with starting symbol $S$)
that uses $n_{A, \alpha}$ times a rule $A \to \alpha$
if and only if
\begin{align*}
\sum_{(S \to \alpha) \in P} n_{S,\alpha}
	&=
1 + \sum_{(A \to \alpha) \in P} |\alpha|_S \cdot n_{A,\alpha}\\
\forall B \in N \setminus \{ S \}
\sum_{(B \to \alpha) \in P} n_{B,\alpha}
	&=
\sum_{(A \to \alpha) \in P} |\alpha|_B \cdot n_{A,\alpha}
\end{align*}
and the underlying graph is connected.
Moreover, if $w$ is generated by such a derivation then
$$
|w|_a = \sum_{(a\to \alpha) \in P} n_{A, \alpha} \cdot |\alpha|_a.
$$
\end{lemma}
Here the underlying graph has $\nonterm$ as vertices and there is an (undirected) edge
$\{A, B\}$ when $n_{A, \alpha} > 0$ and $|\alpha|_B > 0$,
for some $\alpha$.
In~\cite[Thm.~3, 4]{VermaSS05} Verma et al.\ give explicitly a stronger variant of this claim and implicitly formulate Lemma~\ref{lem:derivation_conditions_simple}
in the proof.
Moreover, the condition of the underlying graph being connected is also formulated
via a Presburger arithmetic formula.
While the original construction of such a Presburger formula contained a small error, there are alternative, correct variants in the literature (e.g.~\cite{B06});
we follow those correct constructions.

We extend this characterization to the case of parametric grammars
and reformulate conditions from Lemma~\ref{lem:derivation_conditions_simple}
in terms of mappings, which will allow reasoning on Parikh images and derivations at the same time.
We consider an extension of the Parikh image mappings that also assign counts to pairs $\{s, t\} \times \nonterm$.
A pair $(s, A)$ indicates how many times the nonterminal $A$ is the source of a derivation step, and a pair $(t, A)$ indicates how many times $A$ is introduced by a derivation step.
Thus, we use mappings $f : \{1,\ldots, n\} \cup \{s, t\} \times \nonterm \to \N$.
Let $d = n + 2|\nonterm|$ denote the size of the domain of the mapping.
The mapping restricted to $\{1, \ldots, n\}$ corresponds to the Parikh image.

Given an instantiation $\alpha'$ of production $A \to \alpha$
we denote by $f_{A,\alpha,\alpha'}$ the mapping $f$ with
$f(i) = (P_\Psi(\alpha'))(i)$ for $i = 1, \ldots, n$
(we assume that $P_\Psi$ ignores the nonterminals),
and $f(s, A) = 1$ and $f(s,A') = 0$ for $A \neq A'$
and $f(t, B)$ is $|\alpha|_B$, for each $B \in \nonterm$.
Note that we associate several mappings with a single production,
as there are many instantiations $\alpha'$ for a fixed $\alpha$
and on the other hand several instantiations of a rule can have the same mapping.

We consider multisets $F$ of mappings as above,
so $F : \N ^{\{1,\ldots, n\} \cup \{s, t\} \times \nonterm} \to \N$,
which implicitly defines how many times each mapping $f \in F$ is used;
by $m \cdot \{f\}$ we denote a multiset consisting of $m$ instances of $f$.
By $\sum F$ we denote $\sum_{f \in F} f$, which is an element-wise sum:
$(\sum_{f \in F} f)(i) = \sum_{f \in F} f(i)$.
Given a multiset $F$ of mappings as above we say that a derivation (for some fixed interpretation of parameters)
which uses $n_{A,\alpha, \alpha'}$ times the instantiation $\alpha'$ of rule $A \to \alpha$,
is corresponding to $F$, when
$$
F = \bigcup_{\substack{(A \to \alpha) \in P\\ \alpha'}}
n_{A, \alpha, \alpha'} \cdot \{f_{A,\alpha,\alpha'}\}.
$$

\begin{lemma}
\label{lem:derivation_conditions}
Given a multiset $F$ of mappings there is a derivation corresponding to it if and only if
\begin{subequations}
\label{eq:euler_conditions}
\begin{align}
\left(\sum F\right)(s,S) &= 1 + \left(\sum F\right)(t,S)\\
\left(\sum F\right)(s,A) &= \left(\sum F\right)(t,A) &\text{for all } B \in N \setminus \{S\}
\end{align}
\end{subequations}
and the underlying graph is connected.
Moreover, this derivation yields a word with Parikh image
$
\left(\sum F \right)
$
restricted to $\{1,\ldots, n\}$.
\end{lemma}
Here the underlying graph has nodes $\{A : \left(\sum F\right)(s,A) > 0 \}$
and edges $\{A,B\}$ when there is $f \in F$ such that $f(s,A), f(t,B) > 0$.
\begin{proof}

%

After fixing the parameters and the possible instantiations of the rules,
the parametric CFG becomes a CFG grammar over a finite-size alphabet.
We show that condition~\eqref{eq:euler_conditions} from the Lemma is equivalent to the one from
Lemma~\ref{lem:derivation_conditions_simple}.
Note that the condition that the underlying graph is connected is the same
in proven Lemma and in Lemma~\ref{lem:derivation_conditions_simple}.

Suppose that there is a derivation corresponding to $F$.
Denote by $n_{A, \alpha, \alpha'}$ the counts of the instantiations $\alpha'$ of the rules $A \to \alpha$.
Then we can treat the parametric grammar as an ordinary CFG,
with a pair $(A \to \alpha), \alpha'$ being treated as a single rule $A \to \alpha'$.
Hence the numbers $n_{A, \alpha, \alpha'}$
satisfy the following equations as in Lemma~\ref{lem:derivation_conditions_simple}:
\begin{align*}
\sum_{\substack{(S \to \alpha) \in P\\ \alpha'}} n_{S,\alpha,\alpha'}
	&=
1 + \sum_{\substack{(A \to \alpha) \in P\\ \alpha'}} |\alpha|_S \cdot n_{A,\alpha,\alpha'}\\
\forall B \in N \setminus \{ S \}
\sum_{\substack{(B \to \alpha) \in P\\ \alpha'}} n_{B,\alpha,\alpha'}
	&=
\sum_{\substack{(A \to \alpha) \in P\\ \alpha'}} |\alpha|_B \cdot n_{A,\alpha,\alpha'}\\
\end{align*}
Observe that $\sum_{(A \to \alpha \in P), \, \alpha'} n_{S,\alpha,\alpha'} $
is $(\sum F)(s,A)$
while $\sum_{(A \to \alpha) \in P, \, \alpha'} |\alpha|_B \cdot n_{A,\alpha,\alpha'}$ is $(\sum F)(t, A)$.
Hence the two equations are a reformulation of the conditions~\eqref{eq:euler_conditions}.

In the other direction, the argument is similar,
but note that for a mapping $f \in F$ with count $n_f$
we need to give counts for each production $A \to \alpha$
and its instantiation $\alpha'$ such that $f_{A, \alpha, \alpha'} = f$
and that the sum of those counts is $n_f$. 
This is done arbitrarily: given a mapping $f \in F$ with count $n_f$
we choose arbitrarily a single such rule $A \to \alpha$ and its instantiation  $\alpha'$
such that $f_{A,\alpha, \alpha'} = f$ and set $n_{A, \alpha, \alpha'}$ to $n_f$.
The same argument as above shows that the equations from the Lemma are just reformulations of equations from Lemma~\ref{lem:derivation_conditions_simple}.

The claim on the Parikh image follows:
the derivation uses an instantiation $\alpha'$ of the $A \to \alpha$ rule $n_{A,\alpha,\alpha'}$ times, which yields that the contribution of the letters
in $\alpha'$ is as stated.
\end{proof}

Now the idea is that if $F$ corresponds to a derivation,
then there is a different multiset $F'$ with the same sum $\sum F = \sum F'$.
Moreover, $F'$ satisfies conditions~\eqref{eq:euler_conditions}
but uses fewer different mappings than $F$.
In particular, $F'$ also corresponds to some derivation.

\begin{lemma}
Let a multiset of mappings $F$ satisfy condition~\eqref{eq:euler_conditions}
from Lemma~\ref{lem:derivation_conditions}.
Suppose that there are two multisets $F', F'' \subseteq F$
such that $\sum F' = \sum F''$, $\sum (F \setminus F') \cup F'' = \sum F$
and the underlying graph for $(F \setminus F') \cup F''$ is connected.
Then $(F \setminus F') \cup F''$ satisfies condition \eqref{eq:euler_conditions} from Lemma~\ref{lem:derivation_conditions}.
\end{lemma}
\begin{proof}
As $F' \subseteq F$ the $F \setminus F'$ is well defined and
$\sum ((F \setminus F') \cup F'') = (\sum F) - (\sum F')  + (\sum F'')$.
By assumption $\sum F' = \sum F''$ we have that
$\sum ((F \setminus F') \cup F'') = \sum F$
and hence the assumptions of Lemma~\ref{lem:derivation_conditions} are met
for $(F \setminus F') \cup F''$.
\end{proof}

We now show that given a large enough set of mappings we can always find
two its subsets of the same sum.
Note that here we do not use multisets, but rather simply sets.

\begin{lemma}
\label{lem:subset_sums_equal}
Given a set of $k$ different mappings $f : \{0,1, \ldots , d \} \to \{0,1, \ldots, \ell\}$, where $d > 3$
and $k \geq 2 d \log (d \ell) $, there are two disjoint subsets of this set
that have the same sum.
\end{lemma}
\begin{proof}

It is enough to show that there are two subsets of the same sum,
they can be made disjoint by removing their intersection.

There are $2^k$ different subsets  and at the same time
each value of the sum of mappings is between $0$ and $k \ell$,
so in total there are at most $(k \ell+1)^d$ different possible sums
and so it is enough to show that $(k \ell + 1)^d < 2^k$.
That is $d \log(k \ell+1) < k$,
i.e.\ $d < \frac{k}{\log(k \ell+1)}$.
As $\frac k {\log (k \ell+1)}$ is increasing for $k \geq 2$,
it is enough to verify for $k = 2 d \log (d \ell)$ that :
\begin{equation*}
d < \frac{2 d \log (d \ell)}{\log( 2 d \ell \log (d \ell) +1)}
	=
d \frac{\log ((d \ell)^2)}{\log( 2 d \ell \log (d \ell) +1)} \enspace .
\end{equation*}
Which is equivalent to (setting ($x = d \ell$))
\begin{equation*}
\log (x ^2)> \log(2 x \log x +1)
\end{equation*}
As $\log$ is an increasing function, then it is enough to show that
$x ^2 > 2 x \log x + 1$,
which clearly holds for  $x > 3$, so in particular for $d > 3$.
\end{proof}

\begin{lemma}
\label{lem:small_alphabet}
Given a multiset $F$ of mappings such that there is a derivation corresponding to it,
there is a multiset $F'$ of mappings corresponding to it, $\sum F = \sum F'$ and
$F'$ contains at most $|\nonterm| + 2d \log d \ell$ different mappings,
where $\ell$ is an upper-bound on the value of each mapping in $F$.
\end{lemma}
\begin{proof}
The idea is as follows:
given a multiset $F$ we use Lemma~\ref{lem:subset_sums_equal}
to claim that $F$ uses few different mappings or
to find $F_1, F_2 \subset F$ such that both
$(F \setminus F_2) \cup F_2$ and $(F \setminus F_2) \cup F_1$
correspond to a valid derivation
and $\sum (F \setminus F_1) \cup F_2 = \sum F = \sum (F \setminus F_2) \cup F_1$;
we then replace $F$ with on of them.
To guarantee that the replacement terminates, we introduce a natural well-founded order on multisets of mappings
and show that (at least) one of $(F \setminus F_1) \cup F_2$ and $(F \setminus F_2) \cup F_1$
is strictly smaller in this order than $F$.
This shows that the replacement terminates at some point
and so the resulting multiset uses few different mappings,
which shows the claim of the Lemma.

Consider the \emph{set} of mappings in $F$,
let us linearly order them in some arbitrary way as $f_1, f_2, \ldots, f_m$.
We now treat the multisets of mappings as mappings themselves,
i.e.\ $F'(f_i)$ gives the number of times $f_i$ is in $F'$.
%
We introduce a linear well-founded order $\leq$ on the multisets of mappings:
we first compare by the number of non-zero components
(a smaller number implies the image is smaller according to $\leq$),
i.e.\ if $|\{i\: : \: F'(f_i) > 0 \}| < |\{i\: : \: F''(f_i) > 0 \}|$ then $F' < F''$
and otherwise we compare the mappings lexicographically by coordinates,
i.e.\ if $|\{i\: : \: F'(f_i) > 0 \}| = |\{i\: : \: F''(f_i) > 0 \}|$
and $i$ satisfies $F'(f_j) = F''(f_j)$ for $j < i$ and $F'(f_i) < F''(f_i)$
then $F' < F''$.
By standard arguments, this is a well-founded order
(as the order on $\mathbb N$ is well-founded and a lexicographic order
such that the order on each component is well-founded is itself well-founded).
%
%

We arbitrarily choose $F_0 \subseteq F$ mappings from $F$,
where $|F_0| \leq |\nonterm|$,
which guarantee that the underlying graph is connected;
this can be done as the underlying graph is connected,
by Lemma~\ref{lem:derivation_conditions_simple},
and has $|\nonterm|$ many vertices.
The $F_0$ will be added to the final multiset of mapping to guarantee that the underlying graph is connected.
Let $F' = F \setminus F_0$ be the remaining multiset of mappings.
If $F'$ has more than $2d\log d\ell$ different mappings,
then by Lemma~\ref{lem:subset_sums_equal} there are two disjoint \emph{sets}
$F_1,F_2 \subset F'$ such that $\sum F_1 = \sum F_2$.
Let $F_1' = (F' \setminus F_1) \cup F_2$ and $F_2' = (F' \setminus F_2) \cup F_1$,
where $F_1'$ and $F_2'$ are multisets.
By the same Lemma, $\sum F_1' = \sum F_2' = \sum F'$.
Note also that $F_1', F_2'$ cannot have more non-zero components
than $F'$, as we are only adding mappings which are already in $F'$.
Hence, if $F_1'$ has a non-zero component,
it is also non-zero in $F'$, and similarly for $F_2'$.
Hence $F'$ cannot be smaller than $F_1'$ or $F_2'$ because it has less non-zero components.

%
%
As $F_1 \neq F_2$, consider the smallest $i$ such that $F_1(f_i) \neq F_2(f_i)$.
If $F_1(f_i) > F_2(f_i)$ then $F_1'(f_i) < F'(f_i)$
and $F_1'(f_j) = F'(f_j)$ for $j < i$;
hence $F_1' < F'$ and we replace $F'$ with $F_1'$
(note that it could be that $F_1'$ has fewer non-zero components than $F'$, in which case the above calculations are superfluous).

If $F_1(f_i) < F_2(f_i)$ then similarly $F_2' < F'$.

We can proceed in this manner as long as there are at least $2d\log d \ell$
different mappings in $F'$.
Since $\leq$ is well-founded, the process terminates.
Let $F''$ be the final set;
it uses less than $2d\log (d \ell) = 2(n + 2|\nonterm|)\log((n + 2|\nonterm|)\ell)$ different mappings.
Then $F_0 \cup F''$, where $F_0$ are the initially chosen $|\nonterm|$ many mappings that guarantee connectedness,
we get the desired multiset: the underlying graph is connected thanks to
$F_0$ and $\sum (F_0 \cup F'') = \sum F$ by easy induction
and then condition~\eqref{eq:euler_conditions} from Lemma~\ref{lem:derivation_conditions}
holds and so by this Lemma there is a derivation corresponding to $F_0 \cup F''$.
\end{proof}

This is enough to give a proof of Theorem~\ref{thm:Parikh_image_represented}:
we can construct a formula for theory $T$ combined with quantifier-free Presburger arithmetic representing the conditions needed.
We then finally existentially quantify all intermediate variables to leave the free variables $x_1, \ldots, x_n$ counting the number of times each $\psi_i$ is satisfied.

The formula first guesses the values $\iparam$ of the parameters $\params$ for the grammar.
Then it guesses $|N| + 2d \log d \ell$ mappings
$\{0, 1, \ldots , n\} \cup \{s,t\} \times \nonterm \to \{0, 1, \ldots, \ell\}$,
where $\ell$ is the maximal length of $\Grammar$ rules
(we can represent a mapping with $n+1 + 2|\nonterm|$ integer variables ranging over $\{0, 1, \ldots, \ell\}$).
It also guesses the counts of each mapping,
yielding a multiset $F$.
For each of those mappings it guesses the rule of the grammar and
its instantiation.
It verifies the correctness of those guesses:
\begin{enumerate}
\item That $\sum F$ satisfies \eqref{eq:euler_conditions}.
\label{item:Parikh_image_existence_1}
\item The underlying graph is connected \label{item:Parikh_image_existence_3}
\item That each $f \in F$ indeed corresponds to a guessed rule $A \to \alpha$ and an instantiation $\alpha'$.
\label{item:Parikh_image_existence_2}
\end{enumerate}
If all tests are satisfied then the formula accepts, otherwise it rejects.
Notice that the first condition is a quantifier-free Presburger condition.
For the second condition, for each $f$, we need to guess an assignment to $\locals$.
(We can use a fresh copy of $\locals$ for each of the $|N| + 2d \log d \ell$ mappings).
Via a disjunction over all grammar rules, the formula guesses the corresponding rule, and verifies that the number of symbols in $\alpha$ satisfying $\psi_i$ matches $f(i)$ and similarly for the nonterminal symbols (on both sides).
To see how to do this with a polynomial-size formula, let $\alpha_1 \cdots \alpha_m$ be the right-hand side of the guessed rule.
For each $\psi_i$, introduce a 0-1 variable $c^i_j$ for each $\alpha_j$.
This variable is $0$ if $\psi_i(\alpha_j)$ does not hold and $1$ otherwise.
Then we check $f(i) = \sum_j c^i_j$.
For this, we only require Boolean combinations of the two theories.

Each free-variable $x_i$ of the formula takes the value $\sum_{f \in F} f(i)$ for all $1 \leq i \leq n$.

If there is a satisfying assignment to the formula, giving a mapping $f$, then indeed there is $w \in L_\iparam(\Grammar)$
such that $P_\Psi(w)$ equals $f$ on the first $n$ components.
We guessed the value of the parameters and $F$ corresponds to a derivation,
by Lemma~\ref{lem:derivation_conditions} and the Parikh image of the derived word
is indeed $f = \sum F$ restricted to components $1, \ldots, n$.

In the other direction, if there is $w \in L(\Grammar)$
then by definition for some $\iparam$ we have $w \in L_\iparam(\Grammar)$.
Fix a derivation of $w$ and let $F$ be the corresponding set of mappings;
$F$ satisfies the conditions from Lemma~\ref{lem:derivation_conditions},
so in particular satisfies~\ref{item:Parikh_image_existence_3}.
Then by Lemma~\ref{lem:small_alphabet} we can assume without loss of generality that there are at most $|\nonterm| + 2d \log d \ell$ different mappings in $F$.
From this we can construct assignments to the variables representing $F$ and its counts that satisfy condition~\ref{item:Parikh_image_existence_1}.
Since each $f \in F$ corresponds to a rule in a concrete derivation, it also follows that condition~\eqref{item:Parikh_image_existence_2} is also satisfied.

\subsection{Complexity}

To measure the complexity of Parikh images of parametric grammars, we consider
the problem of deciding whether $P_\Psi(L(\Grammar)) \cap \Sem{\outfmla} = \emptyset$, for
a given existential Presburger formula $\outfmla$ over variables $x_1,\ldots,x_n$.
Here, $\Sem{\outfmla}$ refers to the set of assignments $f : \{1, \ldots, n\} \to \N$
such that $f(1),\ldots,f(n)$ satisfies $F$.
Because we are able to compute a polynomial representation of the Parikh image, the complexity remains similar to $T$.

\begin{theorem}
\label{thm:Parikh_image_exists_T}
Let $T$ be solvable in complexity class $\mathcal C$.
Then the complexity of Parikh over images of parametric grammars $T$ is in $\exists \mathcal C$.
In particular, if $T$ is solvable in \textsc{P}, \np, \pspace then the Parikh images of parametric grammars are in \np, \np, \pspace, respectively.
\end{theorem}

First, we construct $\phi(x_1, \ldots, x_n)$ as in Theorem~\ref{thm:Parikh_image_represented}.
We then need to find an assignment to $x_1, \ldots, x_n$ such that
$\phi(x_1, \ldots, x_n) \land \outfmla(x_1, \ldots, x_n)$ holds.

If there is $w \in P_\Psi(L(\Grammar)) \cap \Sem{\outfmla}$,
fix a derivation of $w$ and let $F$ be the corresponding set of mappings;
$F$ satisfies \eqref{eq:euler_conditions} and $\outfmla$.
For each different $f \in F$ consider its count $n_f$ in $F$.
Consider conditions~\ref{item:Parikh_image_existence_1} and \ref{item:Parikh_image_existence_2} and $\outfmla$.
We claim that if they are satisfied, then they are satisfied for counts $\{n_f\}_{f \in F}$
that are at most exponential.
Given the set of different mappings in $F$,
condition~\eqref{item:Parikh_image_existence_1} and
$\outfmla$
can be interpreted as a system of linear equations on $\{n_f\}_{f \in F}$,
and condition~\eqref{item:Parikh_image_existence_2} requires a guess
of the rule and its instantiation, which can be done separately.
The check whether the underlying graph is connected
can also be done separately, as it does not depend on $\{n_f\}_{f \in F}$.

Then the system of equations is over $|\nonterm| + 2 d \log d \ell$ variables
and with polynomial-size constants.
By standard results, if it has a solution, it has one that is at most exponential-size.
Thus, the integer variables have solutions encodable via a polynomial number of bits.
If $T$ is solvable in P or NP, then we can decide
$P_\Psi(L(\Grammar)) \cap \Sem{\outfmla} \neq \emptyset$
in NP\@.
If $T$ is PSPACE, then the problem is also PSPACE.

\section{Abstraction of String Constraints} \label{sec:string}

In this section we outline an application of the symbolic Parikh image abstraction to solving string constraints.
String constraints arise naturally during symbolic execution analysis of programs, where the string data type is ubiquitous.
During symbolic execution, the potential paths of the program are explored and the constraints on the variables are collected.
For example, the positive branch of an if-statement with the condition $\svar = \cha\chb$ will result in the condition $\svar = \cha\chb$ being added to the collection of constraints on the path.
If an error state is discovered, a check whether the collection of constraints on the path to the state are satisfiable is made.
That is, is there some assignment to the variables that would cause this path to be executed?
Because many of these checks will be made during analysis, it is important that the string constraint solver is efficient.

In this setting, we consider constraints that may contain string- and integer-valued variables, Presburger arithmetic, and a number of common string operations such as concatenation and containment in a regular language.

To help improve solver efficiency, we may consider using the Parikh image abstraction to help identify unsatisfiable constraints and avoid a potentially costly proof search by the solver.
The goal of the abstraction is to overapproximate the satisfiable instances.
This means an unsatisfiable abstracted instance implies the original instance was also unsatisfiable.
That is, we allow false positives, but not false negatives.
We may hope that the abstracted constraint -- which no longer includes string variables but integer variables representing the Parikh image of the strings -- can be solved more quickly than the unabstracted constraint.

In the next sections, we introduce the constraint language, describe our abstraction of the input constraints, then describe a modified representation of the Parikh image of a regular expression.

\subsection{The Constraint Language}

We focus on a subset of the QF\_SLIA theory of SMT-LIB 2.6~\cite{SMTLIB2}.
That is, string constraints with linear integer arithmetic.
The constraint language is explained below.
Because we focus on SMT-LIB, our constraint language only has constraints that a string is in a regular language.
Of course, we can easily extend this language to support containment checks in the language of a parametric context-free grammar.

Note, we assume all formulas are given in negation normal form.
This is because, as explained below, abstracting negative equations can lead to false negatives.
Our constraint language contains the following components.

\begin{itemize}
\item
    String-valued expressions
    \[
        \sexp, \sexp_1, \sexp_2 \defeq
            \word \synalt
            \svar \synalt
            \sconcat(\sexp_1, \sexp_2) \synalt
            \sreplace(\sexp, \sexp_1, \sexp_2) \synalt
            \ssubstr(\sexp, \iexp_1, \iexp_2) \ .
    \]
    That is, string-valued expressions can be
        string literals $\word$,
        string variables $\svar$,
        the concatenation $\sconcat$ of two string-valued expressions $\sexp_1$ and $\sexp_2$,
        the result of replacing the first occurrence of $\sexp_1$ in $\sexp$ with $\sexp_2$, or
        the substring of $\sexp$ from position $\iexp_1$ of length $\iexp_2$,  for integer expressions $\iexp_1$ and $\iexp_2$.

\item
    Boolean-valued string expressions
    \[
        \begin{array}{c}
            \sinre(\sexp, \regex) \synalt
            \neg \sinre(\sexp, \regex) \synalt
            \sexp_1 = \sexp_2 \synalt \\
            \scontains(\sexp_1, \sexp_2) \synalt
            \sprefixof(\sexp_1, \sexp_2) \synalt
            \ssuffixof(\sexp_1, \sexp_2) \ .
        \end{array}
    \]
    That is, the Boolean-valued expressions can be
        the test that $\sexp$ is (or is not) in the regular expression $\regex$,
        that $\sexp_1$ equals $\sexp_2$,
        $\sexp_1$ contains the contiguous substring $\sexp_2$,
        $\sexp_1$ is a prefix of $\sexp_2$, or
        $\sexp_1$ is a suffix of $\sexp_2$.
    We support the regular expressions supported by Z3 which are detailed below.
    Boolean expressions may appear in any positive Boolean combination.

\item
    The integer-valued expressions
    \[
        \slen(\sexp)
    \]
    for string expression $\sexp$.
    Note, other integer valued expressions (that do not include strings) that are supported by Z3 are also permitted.
    In particular, quantifier-free Presburger arithmetic (or linear integer arithmetic).

\item
    Regular expressions have a standard interpretation and are
    \[
        \regex, \regex_1, \regex_2 \defeq
            \cha \synalt
            \rrange(\cha, \cha') \synalt
            \regex_1 \regex_2 \synalt
            \regex_1 \lor \regex_2 \synalt
            \regex_1 \land \regex_2 \synalt
            \neg \regex \synalt
            \regex^\ast \synalt
            \regex^+ \synalt
            \regex^? \synalt
            \regex^{l, h} \synalt
            \emptyset \synalt
            \Sigma
    \]
    where
        $\cha$ is a concrete character;
        $\rrange(\cha, \cha')$ is any character between $\cha$ and $\cha'$ (inclusive);
        $\regex_1 \regex_2$ is concatenation;
        $\regex_1 \lor \regex_2$, $\regex_1 \land \regex_2$, and $\neg \regex$ are Boolean operations;
        $\regex^\ast$ is 0 or more consecutive matches of $\regex$;
        $\regex^+$ is 1 or more consecutive matches of $\regex$;
        $\regex^?$ is 0 or 1 matches of $\regex$;
        $\regex^{l, h}$ is between $l$ and $h$ matches of $\regex$ where $l$ is a nonnegative integer and $h$ is a nonnegative integer or infinity;
        $\emptyset$ matches no word; and
        $\Sigma$ matches any character.
\end{itemize}

\subsection{Abstraction of Input}

Let $\preds := \pred_1, \ldots, \pred_n$ be the predicates of the Parikh image.
We assume that $\pred_1 = \top$ as it is convenient for encoding the length of a string.
We will take a vector view of the Parikh image mappings $f$.
That is, we will denote $f$ as a vector $(f_1, \ldots, f_n)$ where $f_\idxi = f(\idxi)$.
This means we can refer to a vector of variables or expressions that represent a Parikh image.

Let $\anyparikh$ be the Parikh image of the regular expression $\Sigma^\ast$ that matches any string.
We use $\anyparikh$ to assert that any vector of expressions $\vec{\scount} = (\scount_1, \ldots, \scount_n)$ encodes the Parikh image of a string.
E.g., for predicates $\top, \bot$, the counts $0, 1$ are not a valid Parikh image and $\anyparikh(0, 1)$ does not hold.

Our approach creates an overapproximation abstraction of the input SMT-LIB formula.
Each string expression $\sexp$ is abstracted as a vector $\vec{\scount}^\sexp$ where each component $\scount^\sexp_\idxi$ is an expression counting the number of times a character satisfies $\pred_i$ in the value of $\sexp$.
Similarly, regular expressions are abstracted as a formula $\varphi$ recognising the Parikh image of the expression.

This means all string variables $\svar$ are abstracted with a vector $\vec{\svar}^\preds = (\svar^\preds_1, \ldots, \svar^\preds_n)$.
The variable $\svar^\preds_\idxi$ counts the number of occurrences of characters matching the predicate $\pred_\idxi$ in the value assigned to $\svar$.
We assert $\anyparikh(\vec{\svar}^\preds)$ for each abstracted string variable.

Each string expression $\sexp$ is abstracted recursively.
The translation may introduce new variables, for which additional side conditions need to be asserted.
The side conditions are collected as a side-effect of the translation.
Pseudo-code is given in Algorithm~\ref{alg:sexp-translation} and explained below.
The $\Assert$ function adds the assertion to the output formula.

The abstraction of $\word$ directly counts the number of characters of $\word$ that satisfy each predicate.
Since $\word$ is a concrete string, this is a straightforward character-by-character check against each predicate.
For a string variable $\svar$ we use the abstraction variables $\vec{\svar}^\preds$.

The $\sreplace$ operation is the most subtle and reflects the semantics of $\sreplace$ in SMT-LIB.
We introduce fresh variables $\vec{\varsvar}^\preds$ to store the abstracted result of the replace.
There are three possible outcomes.
If $\sexp_2$ does not appear in $\sexp_1$, then $\sexp_1$ is unchanged and we assert $\vec{\varsvar}^\preds = \vec{\scount}^1$ where $\vec{\scount}^1$ is the abstraction of $\sexp_1$.
If $\sexp_2$ appears in $\sexp_1$, then the first instance of $\sexp_2$ is removed from $\sexp_1$ and replaced with $\sexp_3$.
The final case is when $\sexp_2$ is the empty string.
In this case, the SMT-LIB semantics is that $\sexp_3$ is prepended onto $\sexp_1$.
This is encoded by $\vec{\varsvar}^\preds = \vec{\scount}^1 + \vec{\scount}^3$.
Notice, since the values of $\vec{\varsvar}^\preds$ are derived from $\vec{\scount}^1$, $\vec{\scount}^2$, and $\vec{\scount}^3$, we do not need to assert $\anyparikh$ as they are already implied.

Finally, for $\ssubstr$ we again introduce fresh variables $\vec{\varsvar}^\preds$.
We ignore the integer arguments $\iexp_1$ and $\iexp_2$ and simply require that $\vec{\varsvar}^\preds$ is (point-wise) contained within $\vec{\scount}^1$.
We assert $\anyparikh$ to ensure the values represent a string.

\begin{algorithm}
\caption{\label{alg:sexp-translation} $\AbstractSEXP(\sexp)$}
\If{$\sexp = \word$}{
    \Return $\vec{\scount}$ where each $\scount_\idxi$ is the count of characters satisfying $\pred_\idxi$
}
\If{$\sexp = \svar$}{
    \Return $\vec{\svar}^\preds$
}
\If{$\sexp = \sconcat(\sexp_1, \sexp_2)$}{
    \Return{$\AbstractSEXP(\sexp_1) + \AbstractSEXP(\sexp_2)$}
}
\If{$\sexp = \sreplace(\sexp_1, \sexp_2, \sexp_3)$}{
    $\vec{\scount}^1 \gets \AbstractSEXP(\sexp_1)$;
    $\vec{\scount}^2 \gets \AbstractSEXP(\sexp_2)$;
    $\vec{\scount}^3 \gets \AbstractSEXP(\sexp_3)$\;
    Let $\vec{\varsvar}^\preds$ be fresh integer variables\;
    $\Assert(
        \vec{\varsvar}^\preds = \vec{\scount}^1 \lor
        \vec{\varsvar}^\preds =
             \vec{\scount}^1 - \vec{\scount}^2 + \vec{\scount}^3 \lor
        (\vec{\scount}^2 = 0
             \land \vec{\varsvar}^\preds = \vec{\scount}^1 + \vec{\scount}^3)
    )$\;
    \Return $\vec{\varsvar}^\preds$
}
\If{$\sexp = \ssubstr(\sexp_1, \iexp_1, \iexp_2)$}{
    $\vec{\scount}^1 \gets \AbstractSEXP(\sexp_1)$\;
    Let $\vec{\varsvar}^\preds$ be fresh integer variables\;
    $\Assert(\vec{\varsvar}^\preds \leq \vec{\scount}^1)$;
    $\Assert(\anyparikh(\vec{\varsvar}^\preds))$\;
    \Return $\vec{\varsvar}^\preds$
}
\end{algorithm}

We can then abstract Boolean expressions contained in the input.
We first convert each assertion to negation normal form.
This is because we can abstract, for example, $\svar_1 = \svar_2$ but not $\svar_1 \neq \svar_2$.\footnote{
    Because the only satisfying assignments to $\svar_1$ and $\svar_2$ may happen to have the same Parikh image.
    In this case we cannot assert that the Parikh image of $\svar_1$ is not equal to the Parikh image of $\svar_2$, making a satisfiable formula unsatisfiable.
}
We then substitute maximal subexpressions according to the following scheme.

\begin{itemize}
\item
    $\neg \sinre(\sexp, \regex)$
    is replaced by
    $\varphi_{\neg\regex}(\AbstractSEXP(\sexp))$
    where $\varphi_{\neg\regex}$ is the Parikh image of the complement of the language accepted by $\regex$ (discussed in the next section).

\item
    $\sinre(\sexp, \regex)$
    is replaced by
    $\varphi_\regex(\AbstractSEXP(\sexp))$
    where $\varphi_\regex$ is the Parikh image of the language accepted by $\regex$.

\item
    $\sexp_1 = \sexp_2$
    is replaced by
    $\AbstractSEXP(\sexp_1) = \AbstractSEXP(\sexp_2)$.

\item
    $\scontains(\sexp_1, \sexp_2)$
    is replaced by
    $\AbstractSEXP(\sexp_1) \geq \AbstractSEXP(\sexp_2)$.

\item
    $\sprefixof(\sexp_1, \sexp_2)$ and $\ssuffixof(\sexp_1, \sexp_2)$
    are replaced by
    $\AbstractSEXP(\sexp_1) \leq \AbstractSEXP(\sexp_2)$.

\item
    $\slen(\sexp)$
    is replaced by
    $\scount_1$
    where
    $\vec{\scount}$ is the result of $\AbstractSEXP(\sexp)$.
    Recall $\pred_1$ was assumed to be $\top$, so $\scount_1$ gives the length of the string.
\end{itemize}

\subsection{Construction of the Parikh Image}

The abstraction above uses $\varphi_\regex$, which is the symbolic Parikh image abstraction of the regular expression $\regex$.
We describe how we encode $\varphi_\regex$ as an SMT-LIB formula.
Our encoding goes via symbolic automata and is slightly different to the proof of Theorem~\ref{thm:Parikh_image_exists_T}.
The alternative encoding is driven by the concrete transitions of the automaton representing $\regex$, rather than a nondeterministically chosen subset of the states and transitions.
This is because when transitions are represented by variables, many clauses need to contain a disjunction over all possible instantiations of the variables, causing an undesirable polynomial blow-up.
This new encoding may require $2sn\log(n)$ character variables, where $s$ is the number of transitions of the symbolic automaton.
This is theoretically worse than the $2(n + 2|Q|)\log(n + 2|Q|)$ characters needed in the proof of Theorem~\ref{thm:Parikh_image_exists_T}.
In the next section we describe some mitigating optimisations.

Given a regex $\regex$, we build an equivalent symbolic automaton $\aut$.
We briefly recall the definition we use of a symbolic automaton. It is equivalent to parametric context-free grammars where all productions are of the form $(A, \alpha, \varphi(\curr))$ (i.e.\ no parameters) and $\alpha = \curr\ B$ or $\alpha = \epsilon$.

For a given theory $T$, a symbolic automaton is a tuple
$(\controls, \transrel, q_0, F)$
where
    $\controls$ is a finite set of states,
    $\transrel \subseteq \controls \times T(\curr) \times \controls$ is the transition relation,
    $q_0$ is the initial state, and
    $F \subseteq \controls$ is the set of final states.
A run over a word $\cha_1 \ldots \cha_\ell$ is a sequence of transitions
$(q_0, \varphi_1, q_1)(q_1, \varphi_2, q_2)\ldots(q_{\ell-1}, \varphi_\ell, q_\ell)$
where $T \models \varphi_\idxi(\cha_\idxi)$ for all $\idxi$ and $q_\ell \in F$.

Using predicates that are Boolean combinations\footnote{
    Boolean combinations may arise during, e.g., product constructions and complementation.
} of
$\varphi(\curr) := \curr = \cha$
and of
$\varphi(\curr) := \sinre(\curr, \rrange(\cha, \cha'))$,
standard constructions can be used to obtain a symbolic automaton that is equivalent to a regular expression as defined above.
From $\aut$ we build a formula $\varphi(\scount_1, \ldots, \scount_n)$, where the free variable $\scount_\idxi$ indicates the number of characters satisfying predicate $\pred_\idxi$.

Let $\labels(\aut)$ be the set of predicates appearing on the transitions of $\aut$.
We use $\tpred$ to denote these predicates to avoid confusion with the output predicates $\pred_1, \ldots, \pred_n$.
For a predicate $\tpred \in \labels(\aut)$, let $\translabelled(\tpred)$ be the set of transitions of $\aut$ labelled by the predicate $\tpred$.
Let $\labels(\aut)$ be $\tpred_1, \ldots, \tpred_m$.
Using a corrected version of Verma et al.~\cite{VermaSS05,B06}, we can build a linear-sized existential Presburger formula from $\aut$ with free variables $\scount_{\tpred_1}, \ldots, \scount_{\tpred_m}$ where $\scount_{\tpred}$ indicates how many transitions labelled $\tpred$ appear in a run of $\aut$.
That is, given a run $\run = \atran_1 \ldots \atran_\ell$ of $\aut$ and a transition $\atran$, let $\countof{\atran}{\run}$ be the number of occurrences of $\atran$ in $\run$.
Then, for a predicate $\tpred$, let
$\countof{\tpred}{\run} = \sum\limits_{\atran \in \translabelled(\tpred)} \countof{\atran}{\run}$.
If $\varphi(\scount_{\tpred_1}, \ldots, \scount_{\tpred_m})$ holds, then there is a run $\run$ of $\aut$ with $\countof{\tpred_\idxi}{\run} = \scount_{\tpred_\idxi}$ for all $\idxi$.

We split the target counts $\scount_1, \ldots, \scount_n$ of the output predicates $\pred_1, \ldots, \pred_n$ between the counts contributed by each of the transition labels in $\labels(\aut)$.
To do this we introduce variables $\scount^\tpred_\idxi$ indicating that the predicate $\tpred$ accounts for $\scount^\tpred_\idxi$ characters satisfying $\pred_\idxi$.
That is for all $\idxi$
\[
    \scount_\idxi = \sum\limits_{\tpred \in \labels(\aut)} \scount^\tpred_\idxi
    \ .
\]

We then check, for each $\tpred \in \labels(\aut)$ whether there exists a sequence of characters $\cha_1, \ldots, \cha_{\ell'}$ such that $\ell' = \scount_{\tpred}$ (recall $\scount_{\tpred}$ is the number of times $\tpred$ appeared as the label of a transition in the run).
Moreover, for each output predicate $\pred_\idxi$, we require that $\scount^\tpred_\idxi$ is the number of characters $\cha$ in $\cha_1, \ldots, \cha_{\ell'}$ such that $\pred_\idxi(\cha)$ holds.

By Lemma~\ref{lem:subset_sums_equal}, the sequence $\cha_1, \ldots, \cha_{\ell'}$ needs at most $2n\log(n)$ different characters.
We introduce character variables $\cha^\tpred_\idxi$ for $1 \leq \idxi \leq 2n\log(n)$ and character count variables $\chacount^\tpred_\idxi$.
That is, character $\cha^\tpred_\idxi$ appears $\chacount^\tpred_\idxi$ times in $\cha_1,  \ldots, \cha_{\ell'}$.
Additionally, we naturally require that $\cha_\idxi$ satisfies $\tpred$.

All together, we assert (using $\varphi$) that
\begin{itemize}
\item
    the counts for each label ($\scount_\tpred$) is in the Parikh image of the automaton (using $\varphi$),
\item
    the total count for each output predicate ($\scount_\idxi$) is the sum of counts of the labels satisfying the predicates (using $\varphi_\mathrm{sum}$), and
\item
    for each label, we assert (using $\varphi_\mathrm{labels}$) that
    \begin{itemize}
    \item
        the count for that label is spread across $2n\log(n)$ characters (using $\varphi_\mathrm{lcounts}$, recalling $\cha^\tpred_\idxi$ appears $\chacount^\tpred_\idxj$ times),
    \item
        each character satisfies the label (using $\varphi_\mathrm{preds}$), and
    \item
        the counts of the labels satisfying a predicate is correct (using $\varphi_\mathrm{pcounts}$).
    \end{itemize}
\end{itemize}
That is, recalling $\labels(\aut) = \tpred_1, \ldots, \tpred_m$, we define $\varphi_\regex(\scount_1, \ldots, \scount_n)$ to first use the non-symbolic Parikh image
$\varphi(\scount_{\tpred_1}, \ldots, \scount_{\tpred_m})$
to calculate how many time each transition label can occur. Then it uses
$\varphi_\mathrm{sum}$
to assert that the total number of times $\pred_\idxi$ is satisfied is distributed across each transition label.
Finally
$\varphi_\mathrm{labels}$
bridges between the number of labels, the number of characters satisfying those labels, and the number of times each output predicate $\pred_\idxi$ is satisfied.
That is, $\varphi_\regex(\scount_1, \ldots, \scount_n) :=$
\[
    \exists_{\substack{
        \tpred \in \labels(\aut) \\
        1 \leq \idxi \leq n \\
        1 \leq \idxj \leq m
    }}
        \scount_\tpred,
        \scount^\tpred_\idxi,
        \cha^\tpred_\idxj,
        \chacount^\tpred_\idxj \ .\ %
        \varphi(\scount_{\tpred_1}, \ldots, \scount_{\tpred_m})
        \ \land\ %
        \varphi_\mathrm{labels}
        \ \land\ %
        \varphi_\mathrm{sum}
\]
where
\[
    \begin{array}{rcll}
        \varphi_\mathrm{sum} &:=&
            \bigwedge\limits_{1 \leq \idxi \leq n} \left(
                \scount_\idxi =
                \sum\limits_{\tpred \in \labels(\aut)}
                     \scount^\tpred_\idxi
            \right)
        & \mathexpl{5cm}{output counts are split between the transition labels}
        \\
        \varphi_\mathrm{labels} &:=&
        \bigwedge\limits_{\tpred \in \labels(\aut)} \left(
            \varphi_\mathrm{lcounts}
            \land %
            \varphi_\mathrm{preds}
            \land %
            \varphi_\mathrm{pcounts}
        \right)
        & \mathexpl{4.5cm}{correctness of the labels, see below}
        \\
        \varphi_\mathrm{lcounts} &:=&
            \scount_\tpred =
            \sum\limits_{1 \leq \idxj \leq 2n\log(n)}
                \chacount^\tpred_\idxj
        & \mathexpl{5cm}{the count of transitions labelled by $\tpred$ is split between $2n\log(n)$ different characters, each appearing $\chacount^\tpred_\idxj$ times}
        \vspace{1ex}
        \\
        \varphi_\mathrm{preds} &:=&
            \bigwedge\limits_{1 \leq \idxj \leq 2n\log(n)}
                \tpred(\cha^\tpred_\idxj)
        & \mathexpl{5cm}{each of the $2n\log(n)$ characters satisfies $\tpred$}
        \\
        \varphi_\mathrm{pcounts} &:=&
            \bigwedge\limits_{1 \leq \idxi \leq n} \left(
                \scount^\tpred_\idxi =
                \sum\limits_{1 \leq \idxj \leq 2n\log(n)}
                    \chacount^\tpred_\idxj
                    \times
                    \pred^{0,1}_\idxi(\cha^\tpred_\idxj)
            \right)
        & \mathexpl{5cm}{the number of times $\pred_\idxi$ is satisfied by transitions labelled $\tpred$ is the sum of the $2n\log(n)$ character counts that satisfy $\pred_\idxi$}
    \end{array}
\]
and $\pred^{0,1}(\cha)$ is $1$ when $\cha$ satisfies $\pred$ and $0$ otherwise.

\subsection{The Overapproximation}

We put everything together to gain an overapproximation of an SMT-LIB formula containing string and integer expressions.
We abstract each string variable and expression as a sequence of integer variables -- one for each output predicate.
We replace string expressions with their abstracted equivalent, which may include Parikh images of symbolic automata.
If the abstracted constraint is unsatisfiable, we conclude that the original constraint was also unsatisfiable, and avoid reasoning over the string data type.
We describe our experiments in the next section.

\subsection{Extensions}

Here we discuss possible extensions of our constraint language and technique.
First, we could additionally allow context-free constraints, not just regular
expressions. This is allowed already by some string solvers (e.g.\ TRAU
\cite{TRAU,TRAU-pldi}), but this is not yet supported by SMT-LIB. The technique
in this section easily extends to context-free constraints since our general
results in Section \ref{sec:parikh} concerns symbolic context-free grammars.
This can also be extended to symbolic pushdown automata with the restriction
that the number of push symbols in a transition is small. Second, using
parametric grammars, we could
support a currently ``forward-looking feature'' of string theory inside SMT-LIB
2.6, namely, the operator \texttt{to\_re}, which converts a string (possibly
with string variables) into a regular language. This results in a highly
expressive language, which may capture word equations with Kleene stars.
Existing solvers and benchmarks 
only handle the use cases of \texttt{to\_re}, to which the input contains only
string constants. Using parameters, we may capture constraints of the form
$x \in y^*$, where $x$ is a string variable and $y$ is a ``character variable''
(meaning, string variable of length 1). This can be expressed as follows in
SMT-LIB 2.6:
\begin{verbatim}
        (declare-fun x () String)
        (declare-fun w () String)
        (assert (str.in_re w (re.* (str.to_re x) ) ) )
        (assert (str.len x 1))
\end{verbatim}
Third, we
could also allow other effective boolean algebras, and consider instead
\emph{sequence theories}. Although such an extension is partly supported by
leading SMT-solvers like Z3 \cite{Z3} and CVC5 \cite{CVC5}, there is as yet
a standard logic and file format for sequence theories. In addition, the
decidability of such theories has only been very recently studied
\cite{JLMR23}, whereby the quantifier-free fragment consisting of sequence
equational constraints (i.e. concatenation of sequence variables and constants)
and regular constraints (as parametric symbolic automata) is shown to be
reducible to the case of finite alphabet, but incurring an exponential blow-up
in the alphabet size. An example of such a constraint over LIA is
\[
    yz = zy \wedge y \in ([x \equiv_6 p])^+ \wedge z \in ([x \equiv_7 p^+]),
\]
which has solutions $y, z \mapsto \Z^*$, which satisfy the equation $yz = zy$
and that $y$ (resp. $z$) is a sequence of numbers that are $p$ modulo 6 (resp.
$p$ modulo 7), for some $p \in \Z$. Our results allow us to also analyze such
constraints by similar approach outlined above for string constraints, even when
the sequence constraints are additionally extended with other predicates
that we permit for string constraints (e.g. length constraints, contains, etc.)
and symbolic context-free grammars.

\section{Implementation} \label{sec:experiments}

We implemented our approach described in Section~\ref{sec:string} in C++.
We used the Z3~\cite{Z3} library to parse and represent SMT-LIB formulas.
We supported symbolic regular expressions by adapting the symbolic automata code and translations from the Z3 codebase.
In the next sections we describe the optimisations we have implemented, the benchmarks used for testing, and then finally our results and analysis.

\subsection{Optimisations}

We improve the performance of the tool with two optimisations.
The first reduces the number of characters $\cha^\tpred_\idxi$ required for each transition label, the second helps restrict the search space of Z3 when solving the final abstracted constraints.

We remark that our implementation allows predicates
\[
     \pred(\curr) := (\curr = \cha)
     \qquad\text{ and }\qquad
     \pred(\curr) := \sinre(\curr, \rrange(\cha, \cha')) \ .
\]
In our optimisations, we do not exploit this interval representation.
This means our optimisations apply to theories other than strings.
In a dedicated string solver it would be possible to use well-known optimisations for character intervals to produce smaller formulas.

\subsubsection{Reducing the Characters per Transition Label}

In the encoding above, each transition label requires $2n\log(n)$ different character variables.
However, consider the transition predicate $\tpred(\curr) := (\curr = a)$ that asserts that the character on the transition is the `a' character.
Clearly there is only one character that can satisfy $\tpred$ and $2n\log(n)$ characters are not needed.

Similarly, if the only output predicate were $\pred(\curr) := \top$, then all $2n\log(n)$ characters contribute the same vector $(1)$.
In this case also only one character is required.

Using these observations, we approximate the number of characters than can satisfy $\tpred$ while having pairwise different profiles with respect to the output predicates they satisfy.
To do this, we place the output predicates into ``buckets''.
Initially, one may suppose that each predicate is in its own bucket
\[
    \{\pred_1\}, \ldots, \{\pred_n\} \ .
\]
Supposing each character can either satisfy or not satisfy a predicate, a naive upper bound on the number of possible characters with different profiles is $2 \times \cdots \times 2 = 2^n$.

However, suppose the first two output predicates $\pred_1$ and $\pred_2$ were such that there is no character $\cha$ such that $\pred_1(\cha) \land \pred_2(\cha)$ holds.
That is, a character either satisfies $\pred_1$ or $\pred_2$ but never both.
This gives three possibilities ($\cha$ satisfies $\pred_1$, $\cha$ satisfies $\pred_2$, or $\cha$ satisfies neither) instead of the naive upper bound $2^2 = 4$.
In fact, this condition can be tightened: we only need that there is no character satisfying the transition predicate $\tpred$ that can simultaneously satisfy $\pred_1$ and $\pred_2$.

We can extend this to multiple predicates. If $\pred_1$, \ldots, $\pred_{n'}$ are mutually exclusive (i.e.\ any value $\cha$ can only satisfy at most one of the predicates), then there are $n' + 1$ possibilities instead of $2^{n'}$.

When we identify such situations, we can replace the buckets
$\{\pred_1\}, \ldots, \{\pred_{n'}\}$
with a single bucket
$\{\pred_1, \ldots, \pred_{n'}\}$.

Supposing we are able to group the output predicates into buckets $\bucket_1, \ldots, \bucket_{n'}$.
The number of possible vectors with respect to the predicates in a bucket $\bucket$ is the size of the bucket, plus one if it's possible to simultaneous satisfy $\tpred$ and not satisfy any of the predicates in the bucket.
That is, for each bucket, let
\[
    \countof{\tpred}{\bucket} :=
    \begin{cases}
        |\bucket| + 1 & \text{
            if $\tpred(\curr) \land \bigwedge\limits_{\pred \in \bucket} \neg \pred(\curr)$
            is satisfiable
        } \\
        |\bucket| & \text{otherwise.}
    \end{cases}
\]
Our approximation of the upper bound on the number of characters for buckets $\bucket_1, \ldots, \bucket_{n'}$ is then
\[
    \countof{\tpred}{\bucket_1}
    \times \cdots \times
    \countof{\tpred}{\bucket_{n'}} \ .
\]
If this value is less than $2n\log(n)$, we use it instead of $2n\log(n)$ for the characters associated with the transition label $\tpred$ in the encoding above.

To compute the buckets $\bucket_1, \ldots, \bucket_{n'}$ we use Algorithm~\ref{alg:bucket-comp} which is a simple greedy approach to allocating label predicates to buckets.
Note, the ``continue'' keyword jumps to the next iteration of the for loop, so a new bucket is only created if a predicate overlaps with some predicate in all buckets computed so far.

\begin{algorithm}
\caption{\label{alg:bucket-comp} $\ComputeBuckets(\tpred, \{\pred_1, \ldots, \pred_n\})$}
    $\bucketlist \gets \varepsilon$\;
    \For{$\pred$ in $\{\pred_1, \ldots, \pred_n\}$}{
        \For{$\bucket$ in $\bucketlist$}{
            \If{
                $\tpred(\curr) \land \pred(\curr) \land \pred'(\curr)$
                unsatisfiable for all $\pred'$ in $\bucket$
            }{
                $\bucket \gets \bucket \cup \{\pred\}$\;
                continue
            }
        }
        $\bucketlist \gets \bucketlist, \{\pred\}$
    }
    \Return $\bucketlist$
\end{algorithm}

\subsubsection{Restricting the Search Space} \label{sec:restrictsearch}

Our second optimisation is to help Z3 to determine the satisfiability of an abstracted formula.
Suppose for some $\tpred$ we have character variables $\cha^\tpred_1, \ldots, \cha^\tpred_m$.
Suppose further that the solver has managed to determine that the assignment $\cha_1, \ldots, \cha_m$ cannot lead to a satisfying assignment.
It is clear that any permutation of $\cha_1, \ldots, \cha_m$ also cannot lead to a satisfying assignment.
However, without sophisticated inference, the solver needs to repeat the proof for all permutations.
We extend our encoding to eliminate permutations as much as possible.

First, we assume the characters have a linear order $<$.
That is, we can assert $\cha < \cha'$.
In our implementation we represent characters with integers, so such an ordering is readily available.
This means we can add the following constraint to our formula to eliminate permutations.
\[
    \bigwedge\limits_{\tpred \in \labels(\aut)}
        \bigwedge\limits_{1 \leq \idxi < 2n\log(n)}
            \cha^\tpred_\idxi < \cha^\tpred_{\idxi + 1}
\]

We can go a little further an also enforce that characters have different profiles with respect to the satisfaction of the output predicates.
We enforce this with the following constraint.
\[
    \bigwedge\limits_{\tpred \in \labels(\aut)}
        \bigwedge\limits_{1 \leq \idxi < 2n\log(n)}
            \bigvee\limits_{1 \leq \idxj \leq n}
                \pred_\idxj(\cha^\tpred_\idxi) \neq \pred_\idxj(\cha^\tpred_{\idxi + 1})
\]
Notice that we could have enforced this constraint for each pair of characters $\cha_\idxi$, $\cha_{\idxi'}$ for $1 \leq \idxi \neq \idxi' \leq 2n\log(n)$.
However, this would have required a much larger formula.

\subsection{Experimental Results}

Our implementation is written in \cpp\ and uses Z3 4.12.1.
Z3 is used to read SMT-LIB files and its data structures are used to represent the fomulas.
Z3 is also used as the backend solver for the produced constraints.
Our implementation of symbolic automata is a slightly adapted version of the internal Z3 code.\footnote{
    Specifically, we changed the representation to use transition sets instead of vectors to avoid transition duplication.  We also used sets instead of vectors during minterm calculation when complementing automata to avoid multiple copies of the same predicate. Other minor changes include using Z3's push/pop feature instead of reset, and providing some extra convenience functions.
}
We represent characters as Z3 integers, using their unsigned character codes.
The implementation is available via its online repository~\cite{SymParikhGit} and as an artifact with a disk image on Zenodo~\cite{ZenodoArtifact}.

Our tool provides two methods for selecting the predicates to use in the Parikh image.
In the default mode, the predicates are those appearing on the transitions of the symbolic automata constructed when parsing the input regular expressions.
We only select those predicates that are of the form
$\varphi(\curr) := \curr = \cha$
or
$\varphi(\curr) := \sinre(\curr, \rrange(\cha, \cha'))$.
In the second mode, we additionally take predicates of the form
$\varphi(\curr) := \curr = \cha$
from the string literals appearing in the string equations.
That is, if $\cha_1 \ldots \cha_n$ appears as a string literal in a string equation, we introduce the predicates
$\varphi_\idxi(\curr) := \curr = \cha_\idxi$
for all $1 \leq \idxi \leq n$.

We describe the benchmarks used before giving the results.

\subsubsection{Benchmark Sets}

We used several benchmark sets from SMTCOMP 2022~\cite{SMTCOMP2022}, under the QF\_SLIA category.
That is, quantifier-free constraints using strings and linear integer arithmetic.
We also generated a set of benchmarks from regular expressions with a 5-star rating on \url{regexlib.com}.

Because we intend our technique to complement existing solvers, we restricted our attention to ``difficult'' benchmarks.
We defined the ``difficult'' benchmarks to be those that could not be solved by Z3 in less than 10 seconds.

\begin{itemize}
\item
    The Norn benchmarks were introduced for the Norn tool~\cite{AACHRRS14} and consist of concatenations of string literals and variables tested for membership (and non-membership) of regular expressions.
\item
    The Kepler benchmarks were introduced for the Kepler tool~\cite{PLPSQ18} and consist of quadratic word equations.
    That is, equality tests between two concatenations of string literals and variables, with each variable appearing at most twice.
\item
    The WordEQ benchmarks were randomly generated by us for this paper from regular expressions taken from \url{regexlib.com}.
    This website collects user-submitted regular expressions for tasks such as email recognition, currency values, and others.
    We took regular expressions with a 5-star rating to avoid spam submissions.

    The generated benchmarks were designed to test conjunctions of membership queries between overlapping regular expressions.
    We generated 100 benchmarks of the form
    \[
        \begin{array}{c}
            \sinre(\svar_1, (\regex_1 \ldots \regex_n)^\ast)\ \land \\
            \sinre(
                \svar_2,
                (\regex_1 \ldots \regex_n)^+ (\regex'_1 \ldots \regex'_m)^+
            )\ \land \\
            \sinre(\svar_3, (\regex'_1 \ldots \regex'_m)^\ast)\ \land \\
            \sexp_1 = \sexp_2
        \end{array}
    \]
    where $1 \leq n, m \leq 3$ are randomly chosen integers,\footnote{
        Using Python's \texttt{random.randint} function.
    } $\regex_1, \ldots, \regex_n, \regex'_1, \ldots, \regex'_m$ were randomly selected from the regular expressions obtained as above, and $\sexp_1$ and $\sexp_2$ are each concatenations of three variables picked randomly (possibly with duplicates) from $\{\svar_1, \svar_2, \svar_3\}$ such that each variable appears at least once in $\sexp_1$ or $\sexp_2$ (and possibly in both).
\end{itemize}

In addition to the above sets, we also considered other QF\_SLIA benchmarks submitted to SMTCOMP 2022.
These results are not included as the remaining benchmarks either contained unsupported features (such as string-to-integer functions, or character index functions), or were solved within \SI{10}{s} by Z3.

\subsubsection{Comparison Solvers}

We compare our tool with three state-of-the-art string solvers: Z3 (4.12.1)~\cite{Z3}, CVC5 (1.0.5)~\cite{CVC5}, and OSTRICH~\cite{popl19,popl22,ostrich-int,OSTRICH}.
Z3 is a well-known SMT solver developed at Microsoft.
CVC5 performs strongly in SMTCOMP competitions.
Because the performance of Z3 and CVC5 can sometimes be similar, we also compare with two variants of the OSTRICH tool, which uses an automaton-based approach and often out-performed Z3 and CVC5 on unsat instances in SMTCOMP 2022.
The CEA variant of OSTRICH uses cost-register automata and also makes use of Parikh images~\cite{ostrich-int}.
It was run with the parameters \texttt{+parikh} and \texttt{-profile=strings}.
Both variants were taken from the Cea-new branch of OSTRICH, commit ce855e26~\cite{OSTRICH}.

\subsubsection{Results}

Our experiments were performed on a Lenovo X380 Yoga ThinkPad with 8Gb RAM and 8 Intel\textregistered\ i7-8550U 1.8GHz CPUs, running Arch Linux (kernel 6.4.1).
We used the default method for generating predicates for the Parikh image for all benchmarks except Kepler.
For Kepler we extracted predicates from the string literals as the benchmarks do not contain regular expressions.
When running the tools on the ``difficult'' benchmarks, we set the timeout to \SI{30}{s}.

Our tool over-approximates the true satisfiability of the input string equations and may return false-positives.
Hence we are interested in the number of unsatisfiable instances, and those reported incorrectly as sat by our tool.
For each benchmark set we consider the following.
\begin{itemize}
\item
    How many of the total number of benchmarks were ``difficult''?
\item
    How many of the ``difficult'' benchmarks were unsatisfiable instances?
\item
    How many of the unsatisfiable instances were identified as false-positives by our tool (i.e.\ our tool returned ``sat'')?
\item
    The runtime of our solver, compared with other competitive solvers?
\end{itemize}

The results for the final bullet point are presented in Figure~\ref{fig:results}, where our tool is labelled sym\_parikh.
These graphs show the cumulative number of benchmarks solved in the given time.
The remaining data is in Table~\ref{tbl:results}.
We note that Z3 did not solve any of the selected benchmarks within the timeout.
However, we expect this is due to the bias in benchmark selection: we chose ``difficult'' benchmarks, where the Z3 solver was used to determine difficulty.
Hence, only benchmarks that Z3 found difficult were included.

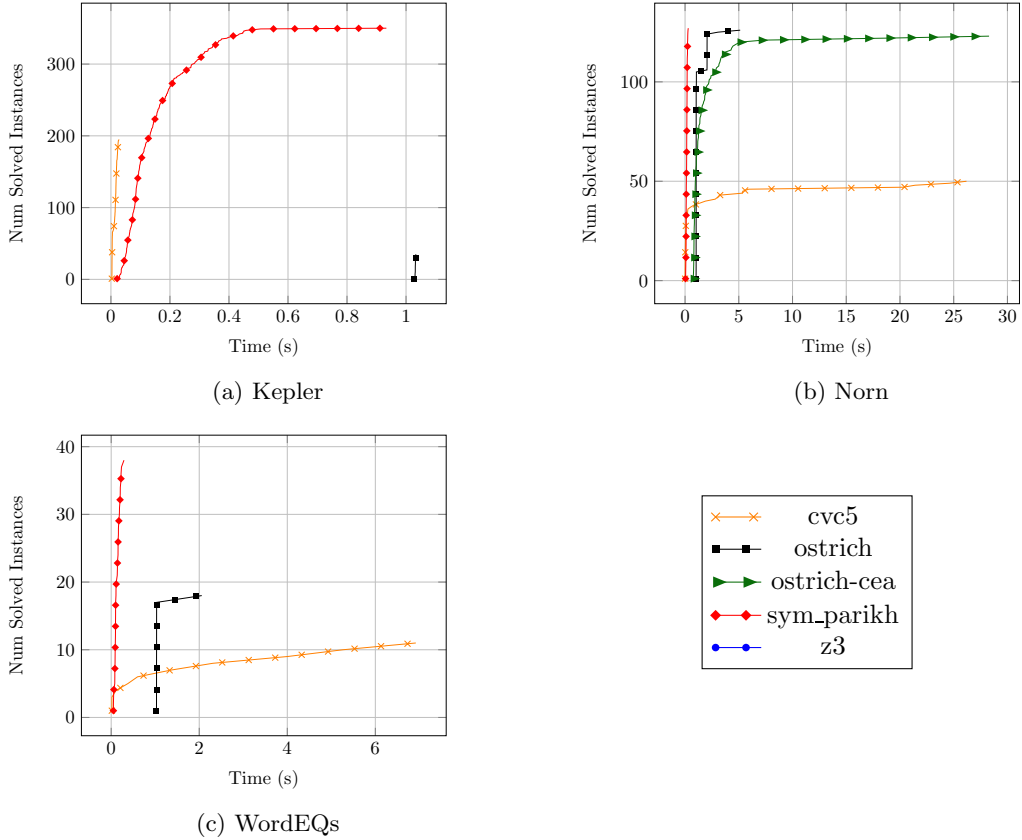
\begin{figure}

\tikzset{cvcdeco/.style={
        orange,
        decoration={
            markings,
            mark = between positions 0 and 1 step 5mm with {
                \node[cross out,draw=orange,inner sep=0mm,minimum size=1mm]{};
            },
        },
        postaction={decorate}
    }
}
\tikzset{ostrichdeco/.style={
        black,
        decoration={
            markings,
            mark = between positions 0 and 1 step 4mm with {
                \node[rectangle,fill=black,inner sep=0mm,minimum size=1mm]{};
            },
        },
        postaction={decorate}
    }
}
\tikzset{ceadeco/.style={
        green,
        decoration={
            markings,
            mark = between positions 0 and 1 step 4mm with {
                \node[isosceles triangle,fill=green,inner sep=0mm,minimum size=1.5mm]{};
            },
        },
        postaction={decorate}
    }
}
\tikzset{symparikhdeco/.style={
        red,
        decoration={
            markings,
            mark = between positions 0 and 1 step 4mm with {
                \node[diamond,fill=red,inner sep=0mm,minimum size=1.5mm]{};
            },
        },
        postaction={decorate}
    }
}
\tikzset{zzzdeco/.style={
        blue,
        decoration={
            markings,
            mark = between positions 0 and 1 step 4mm with {
                \node[circle,fill=blue,inner sep=0mm,minimum size=1mm]{};
            },
        },
        postaction={decorate}
    }
}

\begin{subfigure}{0.48\linewidth}
\popleasychair{\pgfplotstableread{kepler_times.dat}{\table}}{\pgfplotstableread{kepler_times.dat}{\table}}
\begin{tikzpicture}[scale=0.7]
\begin{axis}[
    xlabel={Time (s)},
    ylabel={Num Solved Instances},
    legend pos=north west,
    grid=both
]
\addplot[cvcdeco] table [x = {cvc5_time}, y = {cvc5_count}] {\table};
\addplot[ostrichdeco] table [x = {ostrich_time}, y = {ostrich_count}] {\table};
\addplot[ceadeco] table [x = {ostrich-cea_time}, y = {ostrich-cea_count}] {\table};
\addplot[symparikhdeco] table [x = {sym_parikh_time}, y = {sym_parikh_count}] {\table};
\addplot[zzzdeco] table [x = {z3_time}, y = {z3_count}] {\table};
\end{axis}
\end{tikzpicture}
\caption{Kepler}
\end{subfigure}
\begin{subfigure}{0.48\linewidth}
\popleasychair{\pgfplotstableread{norn_times.dat}{\table}}{\pgfplotstableread{norn_times.dat}{\table}}
\begin{tikzpicture}[scale=0.7]
\begin{axis}[
    xlabel={Time (s)},
    ylabel={Num Solved Instances},
    legend pos=north west,
    grid=both
]
\addplot[cvcdeco] table [x = {cvc5_time}, y = {cvc5_count}] {\table};
\addplot[ostrichdeco] table [x = {ostrich_time}, y = {ostrich_count}] {\table};
\addplot[ceadeco] table [x = {ostrich-cea_time}, y = {ostrich-cea_count}] {\table};
\addplot[symparikhdeco] table [x = {sym_parikh_time}, y = {sym_parikh_count}] {\table};
\addplot[zzzdeco] table [x = {z3_time}, y = {z3_count}] {\table};
\end{axis}
\end{tikzpicture}
\caption{Norn}
\end{subfigure}
\begin{subfigure}{0.48\linewidth}
\popleasychair{\pgfplotstableread{wordeqs_times.dat}{\table}}{\pgfplotstableread{wordeqs_times.dat}{\table}}
\begin{tikzpicture}[scale=0.7]
\begin{axis}[
    xlabel={Time (s)},
    ylabel={Num Solved Instances},
    legend pos=north west,
    grid=both
]
\addplot[cvcdeco] table [x = {cvc5_time}, y = {cvc5_count}] {\table};
\addplot[ostrichdeco] table [x = {ostrich_time}, y = {ostrich_count}] {\table};
\addplot[ceadeco] table [x = {ostrich-cea_time}, y = {ostrich-cea_count}] {\table};
\addplot[symparikhdeco] table [x = {sym_parikh_time}, y = {sym_parikh_count}] {\table};
\addplot[zzzdeco] table [x = {z3_time}, y = {z3_count}] {\table};
\end{axis}
\end{tikzpicture}
\caption{WordEQs}
\end{subfigure}
\begin{subfigure}{.48\linewidth}
\begin{center}
\begin{tikzpicture}
\begin{axis}[
   hide axis,
   xmin=50,
   xmax=50,
   ymin=50,
   ymax=50,
   legend style={at={(0.3,0.6)},anchor=west}
]
\addplot[cvcdeco] {x};
\addplot[ostrichdeco] {x};
\addplot[ceadeco] {x};
\addplot[symparikhdeco] {x};
\addplot[zzzdeco] {x};
\legend{cvc5,ostrich,ostrich-cea,sym\_parikh,z3}
\end{axis}
\end{tikzpicture}
\end{center}
\end{subfigure}
\caption{\label{fig:results}
    The number of instances solved in the given time across the benchmark sets.
    The line markings (shapes) are only to distinguish lines without colors, and are not individual data points.
    Our tool is labelled sym\_parikh.
}
\end{figure}

\begin{table}
\caption{\label{tbl:results}
     Summary data for benchmark sets.
     For each benchmark set, the table shows
        the number of instances,
        how many of the instances are ``difficult'',
        the number of difficult benchmarks that are unsatisfiable,
        and the number of those unsatisfiable instances identified as satisfiable by our approximation.}
\begin{tabular}{lccccc}
\toprule
Set & Instances & Difficult & Unsat & Approx Sat \\
\midrule
Kepler & 587 & 350 & 195 & 129 \\
Norn & 1027 & 127 & 127 & 9 \\
WordEQs & 100 & 38 & 35 & 0 \\
\bottomrule
\end{tabular}
\end{table}

\subsection{Analysis}

Of 515 difficult instances, the majority -- 357 (\SI{69}{\%}) -- were found to be unsatisfiable.
Of these, 219 (\SI{61}{\%}) were able to be proved unsatisfiable using the Parikh image abstraction.
Thus, our approach gives useful results in \SI{42}{\%} of considered cases.
It can be seen in Figure 1 that performance was relatively robust on our benchmarks when compared with with the exact analysis of the comparison tools, with all queries answered within \SI{1}{s}.
This shows some potential for the use of the approach in the optimisation of string solvers.
However, the performance can be seen to vary between the benchmarks sets.
Since the summary results can be affected by the number of available benchmarks in each set, we discuss each set individually below.
This will allow us to gain some intuition on where the approach may be best applied, and where it may be less useful.

Our tool performed well on the Norn benchmarks, solving all instances almost immediately.
All difficult benchmarks were unsatisfiable instances.
Our tool reported 9 false positives.
These results are promising and indicate that the Parikh image abstraction may prove useful in quickly filtering unsatisfiable string constraints with regular expression containment checks.

The performance on the Kepler benchmarks was more mixed.
These benchmarks proved difficult for most solvers, with only CVC5 and our tool able to return a large number of answers within the timeout period.
CVC5 solved fewer instances than our tool, but did so more quickly.
We note that our tool solved almost all instances within \SI{0.5}{s}.
However, out of 195 unsatisfiable instances, our tool reported 129 false positives.
We conjecture that this high false positive rate can be explained by the nature of the word equations.
The values taken by the string variables were not limited by regular expression containment checks.
This provides a lot of freedom for variables to take on values that equalise the Parikh images of both sides of the word equations, especially in cases where a variable only appears on one side of the equation.
For example, in $\cha \svar = \svar \varsvar$, the variable $\varsvar$ can contain the required $\cha$ character.
Unsatisfiability of Parikh image equality can require rarer inconsistencies.
For example $\svar \varsvar = \varsvar \cha \svar$ will always require one more $\cha$ character on the right hand side than the left.

Finally, our tool performed well on the difficult WordEQ instances.
Of the 38 that were difficult for Z3, 35 were unsatifiable instances.
Our tool was able to determine the correct result quickly in all cases.
This shows that the Parikh image abstractions may prove useful for examples containing complex interactions between overlapping regular expressions.

\section{Conclusion}
\label{sec:conc}



We have investigated Parikh images of languages over symbolic and parametric alphabets.
In such a settings, the large, or even infinite alphabet makes a naive use of Parikh images impractical.
Instead, our parametric version of the Parikh image is relative to a sequence of predicates $\Psi = \psi_1, \ldots, \psi_n$ and counts the number of times each predicate is satisfied by a character in the word.

The fact that Parikh images over classical context-free grammars can be computed by a linear-sized existential Presburger formula is a key ingredient in several verification applications.
We introduce a parametric version of context-free grammars and an equivalent pushdown model.

Because the alphabet is large and multiple predicates can be satisfied simultaneously, one may expect an exponential blow-up over the classical results.
Surprisingly, this turns out not to be the case.
We can represent the Parikh image of a parametric context-free grammar with a polynomially-sized existential formula, and the complexity of related decision problems remains the same.

We presented an application of our results to overapproximate satisfiability of string constraints and provided an implementation based on Z3.
Our experimental results showed that constraints that are difficult for existing solvers can be solved quickly using our abstraction.

\textit{Future work.} These initial results suggest several avenues of future 
work.
We first discuss limitations of our implementation. 
Firstly, our implementation makes a naive selection of
predicates $\Psi$ over which to compute the Parikh image.
Improved predicate selection algorithms may balance the need for insightful information about the constraints being analysed, and the need to keep the number of variables small to allow constraints to be solved quickly.
One may also investigate how existing solvers can deploy these techniques from
within the solver, rather than as a one-off preprocessing step that analyses the
whole formula at once. Secondly, our prototypical application to string
constraint solving does not exploit the full potential of the results. For
example, SMT-LIB does not currently support context-free constraints (as
supported by some solvers like TRAU \cite{TRAU,TRAU-pldi}) and
sequence theories over any effective boolean algebra \cite{JLMR23}, and we
have remarked that our results admit an easy extension to these.
As an example, we may use parametric context-free languages to analyse streams of XML data, where the set of possible tags is infinite and should respect a nested structure.

On the theory side, it is still an open problem whether our results can be
extended to other classes of recognizers over infinite alphabets, e.g.,
\cite{SRA,BS19,nominal-automata,FL22}. 
In
particular, we mention recent results on Parikh images of subclasses of nominal
automata \cite{parikh-infinite} and variants of data automata \cite{FL22}, which
provide a more precise Parikh abstraction and thus require a higher computational
complexity (e.g. in \cite{FL22} double-exponential time algorithms).
Secondly, in the light of the polynomial-time complexity result
\cite{KT10} on reasoning about Parikh images of NFA with fixed alphabet size $k$, one
could study Parikh images of parametric automata with a fixed number of
predicates (for certain alphabets like ASCII, this number might be as small as
10 \cite{pldi23,DV21}). Here, a simple application of the result in \cite{KT10} yields a
polynomial-time complexity for any fixed $k$, but the actual complexity would be
double exponential in $k$. Is it possible to lower this to a single 
exponential in $k$?

Finally, one could investigate further potential applications of our results.
For example, as explained in \cite{DV21}, model checking is typically done over
Kripke structures over atomic propositions $P_1,\ldots,P_n$. This gives rise as
well to exponential-sized alphabets. Parikh's Theorem for symbolic automata
could potentially be used to model checking temporal logics with additional
predicate-counting abstractions. Similar applications for the case of finite
alphabets have been discussed in \cite{HL11,counting-CTL,counting-LTL}. To
avoid potentially large automata, one could potentially also consider
restrictions of temporal logics (e.g. LTL with only future/global operators
\cite{BLW13}).

\OMIT{
e.g., model checking of 
Our results may also find applications in LTL(F) model-checking.
By creating $n$ copies of the alphabet, is possible to split the computation into $n$ phases that can be tracked by the Parikh image.
This information could be used to verify LTL formulas with $n$ nestings of the F operator.
}


\section*{Data Availability Statement}
The implementation is available for reuse via its online repository~\cite{SymParikhGit}.
It is also available for reproduction as an artifact with a disk image on Zenodo~\cite{ZenodoArtifactv3}, version v3.

\paragraph{Acknowledgments}
We thank anonymous reviewers, Nikolaj Bjorner, Oliver Markgraf, and Margus Veanes for helpful comments.
Artur Je\.z was supported under National Science Centre, Poland project number~2017/26/E/ST6/00191.
Anthony Lin was supported by European Research Council under European Union’s Horizon research and innovation programme (grant agreement no 101089343).

\section{Errata} \label{sec:errata}

There is an error in the optimisation proposed in Section~\ref{sec:restrictsearch} and published in POPL 2024.
It may happen, for example, that a transition label permits only one possible value.
It is then not possible to choose values for $\cha^\tpred_\idxi$ such that $\cha^\tpred_\idxi < \cha^\tpred_{\idxi+1}$.
This may result in some satisfiable instances becoming unsatisfiable.
To correct the optimisation, we can apply it only to characters $\cha^\tpred_\idxi$ such that $\chacount^\tpred_\idxj > 0$.
That is, we require a character differs from its predecessor, or it is not used at all.
If a character is not used at all, we can assert it is the same as its predecessor to further restrict the search space.
We can further assert that if a character is not used, then none of its successors are used either.

The two optimising formulas become
\[
    \bigwedge\limits_{\tpred \in \labels(\aut)}
        \bigwedge\limits_{1 \leq \idxi < 2n\log(n)}
            (\chacount^\tpred_{\idxi+1} = 0
             \land
             \cha^\tpred_\idxi = \cha^\tpred_{\idxi + 1})
            \lor
            \cha^\tpred_\idxi < \cha^\tpred_{\idxi + 1}
\]
and
\[
    \bigwedge\limits_{\tpred \in \labels(\aut)}
        \bigwedge\limits_{1 \leq \idxi < 2n\log(n)} \left(
            \chacount^\tpred_{\idxi+1} = 0
            \lor
            \bigvee\limits_{1 \leq \idxj \leq n}
                \pred_\idxj(\cha^\tpred_\idxi) \neq \pred_\idxj(\cha^\tpred_{\idxi + 1})
        \right)
\]
with additionally
\[
    (\chacount^\tpred_\idxi = 0) \Rightarrow (\chacount^\tpred_{\idxi + 1} = 0) \ .
\]

With the new encoding, the experimental results are changed.
Updated runtimes are given in Figure~\ref{fig:results-updated} and summary data is given in Table~\ref{tbl:results-updated}.
The rerun experiments were conducted on the same Lenovo X380 Yoga ThinkPad with 8Gb RAM and 8 Intel\textregistered\ i7-8550U 1.8GHz CPUs, running Arch Linux, updated to kernel 6.7.3.
We compared with Z3 version 4.12.5 (updated from 4.12.1) and CVC5 1.1.1 (updated from 1.0.5).
The same version of OSTRICH and OSTRICH-CEA were used.

The change to the Z3 version used means there are slight differences in how many benchmarks were considered ``hard''.
In particular, 310 Kepler benchmarks were considered hard, whereas 350 were difficult for Z3 previously.
Our estimation of the number of unsatisfiable instances has also changed in some cases.
These changes will be for benchmarks where our tool was the only tool returning a result, which was erroneously unsat instead of sat due to the over-aggressive optimisation.

The largest difference in the results is for the Norn benchmark set.
In this case, the number of false positives reported by our tool increased from 9 out of 127 to 110 out of 126 benchmarks.
The other two benchmark sets showed only slight changes to the overall picture of the results.

Overall, there are now 474 difficult instances, with 313 being unsatisfiable (66\%).
Of these unsatisfiable benchmarks, our tool detected unsatisfiability in 72 cases (23\%).
Thus, our approach gives useful results in 15\% of considered cases.
This is a drop from the 42\% reported previously.

However, we observe that runtimes remain good for all benchmarks, meaning that our technique can quickly be used as an unsatisfiability check when traditional techniques are taking a long time.
Thus, the approach remains a viable approximate filter, particularly for benchmarks of the shape explored by the WordEQs benchmarks, which are of the form that initially motivated this work.

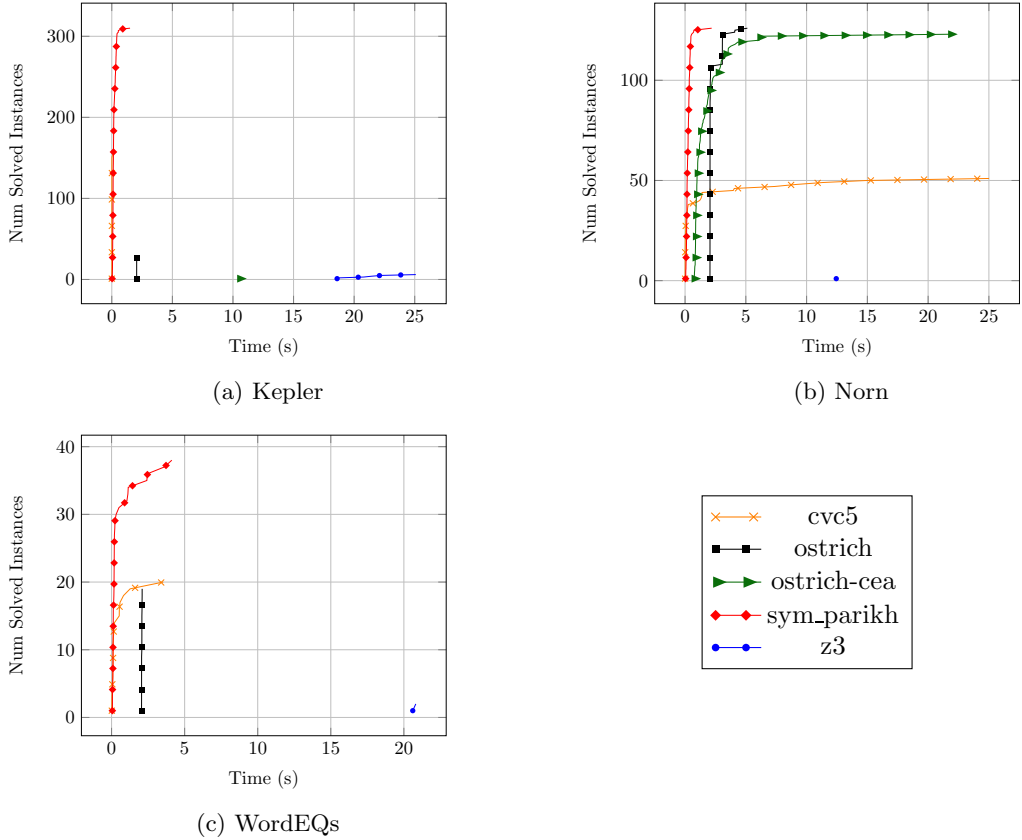
\begin{figure}

\tikzset{cvcdeco/.style={
        orange,
        decoration={
            markings,
            mark = between positions 0 and 1 step 5mm with {
                \node[cross out,draw=orange,inner sep=0mm,minimum size=1mm]{};
            },
        },
        postaction={decorate}
    }
}
\tikzset{ostrichdeco/.style={
        black,
        decoration={
            markings,
            mark = between positions 0 and 1 step 4mm with {
                \node[rectangle,fill=black,inner sep=0mm,minimum size=1mm]{};
            },
        },
        postaction={decorate}
    }
}
\tikzset{ceadeco/.style={
        green,
        decoration={
            markings,
            mark = between positions 0 and 1 step 4mm with {
                \node[isosceles triangle,fill=green,inner sep=0mm,minimum size=1.5mm]{};
            },
        },
        postaction={decorate}
    }
}
\tikzset{symparikhdeco/.style={
        red,
        decoration={
            markings,
            mark = between positions 0 and 1 step 4mm with {
                \node[diamond,fill=red,inner sep=0mm,minimum size=1.5mm]{};
            },
        },
        postaction={decorate}
    }
}
\tikzset{zzzdeco/.style={
        blue,
        decoration={
            markings,
            mark = between positions 0 and 1 step 4mm with {
                \node[circle,fill=blue,inner sep=0mm,minimum size=1mm]{};
            },
        },
        postaction={decorate}
    }
}

\begin{subfigure}{0.48\linewidth}
\popleasychair{\pgfplotstableread{kepler_times_20240209.dat}{\table}}{\pgfplotstableread{kepler_times_20240209.dat}{\table}}
\begin{tikzpicture}[scale=0.7]
\begin{axis}[
    xlabel={Time (s)},
    ylabel={Num Solved Instances},
    legend pos=north west,
    grid=both
]
\addplot[cvcdeco] table [x = {cvc5_time}, y = {cvc5_count}] {\table};
\addplot[ostrichdeco] table [x = {ostrich_time}, y = {ostrich_count}] {\table};
\addplot[ceadeco] table [x = {ostrich-cea_time}, y = {ostrich-cea_count}] {\table};
\addplot[symparikhdeco] table [x = {sym_parikh_time}, y = {sym_parikh_count}] {\table};
\addplot[zzzdeco] table [x = {z3_time}, y = {z3_count}] {\table};
\end{axis}
\end{tikzpicture}
\caption{Kepler}
\end{subfigure}
\begin{subfigure}{0.48\linewidth}
\popleasychair{\pgfplotstableread{norn_times_20240209.dat}{\table}}{\pgfplotstableread{norn_times_20240209.dat}{\table}}
\begin{tikzpicture}[scale=0.7]
\begin{axis}[
    xlabel={Time (s)},
    ylabel={Num Solved Instances},
    legend pos=north west,
    grid=both
]
\addplot[cvcdeco] table [x = {cvc5_time}, y = {cvc5_count}] {\table};
\addplot[ostrichdeco] table [x = {ostrich_time}, y = {ostrich_count}] {\table};
\addplot[ceadeco] table [x = {ostrich-cea_time}, y = {ostrich-cea_count}] {\table};
\addplot[symparikhdeco] table [x = {sym_parikh_time}, y = {sym_parikh_count}] {\table};
\addplot[zzzdeco] table [x = {z3_time}, y = {z3_count}] {\table};
\end{axis}
\end{tikzpicture}
\caption{Norn}
\end{subfigure}
\begin{subfigure}{0.48\linewidth}
\popleasychair{\pgfplotstableread{wordeqs_times_20240209.dat}{\table}}{\pgfplotstableread{wordeqs_times_20240209.dat}{\table}}
\begin{tikzpicture}[scale=0.7]
\begin{axis}[
    xlabel={Time (s)},
    ylabel={Num Solved Instances},
    legend pos=north west,
    grid=both
]
\addplot[cvcdeco] table [x = {cvc5_time}, y = {cvc5_count}] {\table};
\addplot[ostrichdeco] table [x = {ostrich_time}, y = {ostrich_count}] {\table};
\addplot[ceadeco] table [x = {ostrich-cea_time}, y = {ostrich-cea_count}] {\table};
\addplot[symparikhdeco] table [x = {sym_parikh_time}, y = {sym_parikh_count}] {\table};
\addplot[zzzdeco] table [x = {z3_time}, y = {z3_count}] {\table};
\end{axis}
\end{tikzpicture}
\caption{WordEQs}
\end{subfigure}
\begin{subfigure}{.48\linewidth}
\begin{center}
\begin{tikzpicture}
\begin{axis}[
   hide axis,
   xmin=50,
   xmax=50,
   ymin=50,
   ymax=50,
   legend style={at={(0.3,0.6)},anchor=west}
]
\addplot[cvcdeco] {x};
\addplot[ostrichdeco] {x};
\addplot[ceadeco] {x};
\addplot[symparikhdeco] {x};
\addplot[zzzdeco] {x};
\legend{cvc5,ostrich,ostrich-cea,sym\_parikh,z3}
\end{axis}
\end{tikzpicture}
\end{center}
\end{subfigure}
\caption{\label{fig:results-updated}
    The number of instances solved in the given time across the benchmark sets.
    The line markings (shapes) are only to distinguish lines without colors, and are not individual data points.
    Our tool is labelled sym\_parikh.
}
\end{figure}

\begin{table}
\caption{\label{tbl:results-updated}
     Summary data for benchmark sets.
     For each benchmark set, the table shows
        the number of instances,
        how many of the instances are ``difficult'',
        the number of difficult benchmarks that are unsatisfiable,
        and the number of those unsatisfiable instances identified as satisfiable by our approximation.}
\begin{tabular}{lccccc}
\toprule
Set & Instances & Difficult & Unsat & Approx Sat \\
\midrule
Kepler & 587 & 310 & 154 & 127 \\
Norn & 1027 & 126 & 126 & 110 \\
WordEQs & 100 & 38 & 33 & 4 \\
\bottomrule
\end{tabular}
\end{table}

\bibliographystyle{ACM-Reference-Format}
\bibliography{refs}

\end{document}